%% file: main.tex
\documentclass[letterpaper,twocolumn,10pt]{article}
\usepackage{usenix}

\input{header}

\begin{document}

\date{}

\title{ZipPIR: High-throughput Single-server PIR without Client-side Storage}

\author{
{\rm Rasoul Akhavan Mahdavi}
\and
{\rm Abdulrahman Diaa}
\and
{\rm Florian Kerschbaum}
\and
University of Waterloo\\
\{rasoul.akhavan.mahdavi, abdulrahman.diaa, florian.kerschbaum\}@uwaterloo.ca
} 

\maketitle

%
\begin{abstract}
  Private Information Retrieval (PIR) allows a client to privately access a database without revealing which element is accessed.
  Initial PIR protocols based on Ring Learning with Errors (RLWE) demonstrated the practicality of PIR, but achieve limited throughput.
  Alternatively, high-throughput protocols leverage an offline phase that requires substantial client-side storage (e.g., hints in SimplePIR) or involve prohibitive communication costs during the offline phase (e.g., Piano).
  These limitations conflict with the practical constraints of resource-limited clients and are further exacerbated by dynamic databases, where updates necessitate costly regeneration and retransmission of hints.

  To address these challenges, we propose \protocol{}, a high-throughput PIR protocol that compresses LWE ciphertexts into significantly smaller Paillier ciphertexts.
  \protocol{} leverages the offline phase to obtain this size reduction without incurring the associated computational cost in the online phase.
  Moreover, under computational assumptions, \protocol{} features an almost silent offline phase, requiring no communication beyond an initial public key, enabling the server to independently generate and update hints during idle times without client interaction.
  \protocol{} achieves over 2 GB/s of throughput — comparable to state-of-the-art protocols such as SimplePIR — without the need for a large client-stored hint.
  For PIR over a 1 GB database, \protocol{} has up to 10x higher throughput than existing protocols with no client-side storage, while requiring less than 200 KB of server-side storage per client, significantly enhancing scalability for practical deployments.
  While prior PIR protocols using Paillier are very inefficient, \protocol{} is the first PIR protocol using Paillier that achieves throughput that is competitive with state-of-the-art PIR protocols.
\end{abstract}

\input{introduction}
\input{background}
\input{constructions}
\input{zippir}
\input{evaluation}

\input{discussion}
\input{applications}
\input{related-work}
\input{related-work-pir}

\section{Conclusion}
In this work, we proposed \protocol{}, a client-efficient single-server PIR protocol.
\protocol{} is designed to minimize client-side storage requirements, while maintaining high throughput.
After sending an initial setup, the offline phase can be performed silently without any interaction with the client or client-side computation or storage.
By delegating all expensive offline computation and storage to the server, database updates can be handled silently, without communication with all clients.
At the heart of our construction is a new technique for compressing LWE ciphertexts into extremely small additive ciphertexts.
Our compression technique may be of independent interest and is applicable in many applications which send (R)LWE ciphertexts over the network, not only PIR.
Our evaluation shows that our compression can achieve up to 99\% size reduction for LWE ciphertexts.
Moreover, our evaluation of our \protocol{} shows how we can achieve over 3 GB/s of throughput in the online phase, without imposing any client-side burden.
Our works demonstrates that PIR using additive schemes such as Paillier, which were thought to be inefficient, can be practical and competitive with PIR schemes which are purely based on lattice-based schemes.

\section*{Ethical Considerations}

In this work, the main goal is to improve the efficiency of Private Information Retrieval (PIR), a privacy-enhancing technology that enables users to retrieve data without revealing their specific items of interest. We have identified several key stakeholders and evaluated the potential impacts of our research on each:

\paragraph{Potential Stakeholders}
\begin{itemize}\itemsep0mm
    \item \textbf{End-Users of digital services.} The primary stakeholders are individuals seeking to use digital services. By making PIR easier to deploy in real-world applications (such as password leak checks or SCT auditing), this research empowers users to verify their security status without surrendering sensitive metadata to service providers. The benefit to these users is removing the trade-off between utility and privacy.
    \item \textbf{Service Providers.} Improving efficiency lowers the computational and financial barriers for organizations to adopt privacy-preserving protocols.
    \item \textbf{Potential Misusers (Negative Stakeholders).} We acknowledge that PIR may allow bad actors to attempt to access data while using the protocol’s privacy properties to mask illegal activities.
\end{itemize}

\paragraph{Balancing Benefits and Risks}
Our decision to conduct and publish this research was motivated by the significant barriers currently preventing the widespread adoption of PIR. We weighed the potential for misuse against the tangible benefits of empowering users to remain safe on the internet while maintaining their privacy. Because PIR only protects the query and not the authorization (i.e., it does not bypass access control mechanisms already in place by the data provider), we believe the risk of masking illegal activity is mitigated by existing security layers.

Ultimately, we believe the potential gains—evidenced by current real-world deployments such as Password Checkup—outweigh the potential negative outcomes. We are unaware of any further ethical considerations or overlooked stakeholders regarding the improved efficiency of PIR.

\section*{Open Science}

We open-source our implementation of our proposed protocols and the scripts used to perform our experimental evaluation such that it may be evaluated for both functionality and reproducibility.~\footnote{\url{https://zenodo.org/records/17907225}}

\bibliographystyle{plain}
\bibliography{references}

\input{appendix}

\end{document}

%% file: header.tex
\usepackage{algorithm}
\usepackage[noend]{algpseudocode}

\makeatletter
\renewcommand{\ALG@beginalgorithmic}{\scriptsize}

\makeatother

\usepackage{booktabs}
\usepackage{amsfonts, mathtools}
\usepackage{amsthm}
\usepackage{ifthen}
\usepackage{amsmath}
\usepackage[lambda ,advantage, operators, sets , adversary, landau, probability, notions, logic, ff , mm, primitives, events, complexity, asymptotics, keys]{cryptocode}
\usepackage{xurl}
\usepackage{multirow}
\usepackage{multicol}
\usepackage{graphicx}
\usepackage{rotating}
\usepackage{colortbl}
\usepackage{makecell}
\usepackage{tikz}
\usepackage{caption}
\usepackage{subcaption}
\usepackage{diagbox}
\usepackage[first=0,last=9]{lcg}
\usepackage{stfloats}
\usepackage{diagbox}
\usepackage{thmtools}
\usepackage{adjustbox}
\usepackage{xcolor}

\usepackage{pgfplots}
\pgfplotsset{compat=1.15}

\newcommand{\diaa}[2][]{\ifthenelse{\equal{#1}{}}{}{ \textcolor{red}{\sout{#1}}} \textcolor{blue}{D: #2}}
\renewcommand{\fk}[2][]{\ifthenelse{\equal{#1}{}}{}{ \textcolor{red}{\sout{#1}}} \textcolor{orange}{FK: #2}}
\newtheorem{theorem}{Theorem}
\newtheorem{definition}{Definition}
\newtheorem{proposition}{Proposition}
\newtheorem{corollary}{Corollary}

\newcommand{\round}[1]{\lfloor #1 \rceil}
\newcommand{\db}{\pckeystyle{db}}
\newcommand{\qu}{\pckeystyle{qu}}
\newcommand{\ck}{\pckeystyle{ck}}
\newcommand{\eck}{\pckeystyle{eck}}
\newcommand{\pck}{\pckeystyle{pck}}
\newcommand{\ct}{\pckeystyle{ct}}
\newcommand{\pt}{\pckeystyle{pt}}
\newcommand{\ans}{\pckeystyle{ans}}

\newcommand{\hint}{\textbf{H}}

\newcommand{\protocol}{ZipPIR}

\newcommand{\addkey}{\sk_{A}}
\newcommand{\paillierkey}{\sk_{P}}
\newcommand{\paillierpk}{\pk_{P}}
\newcommand{\lwekey}{\sk}

\newcommand{\mask}{\textbf{A}}

\newcommand{\appendixsectionname}{Appendix}
\newcommand{\appsection}[1]{\appendixsectionname~\ref{#1}}

\usepackage{stackengine}
\stackMath

\usepackage{cleveref}

\definecolor{MyGreen}{RGB}{184,223,184}
\definecolor{MyRed}{RGB}{244,187,187}

%% file: introduction.tex
    \section{Introduction}

\begin{table}[h!]
    \caption{Performance of PIR over a 1 GB database for different protocols. 
    $*$ Results for a 1 GB database are extrapolated from the paper. 
    $\dagger$ The only client-side storage is a long-term secret key.}
    \label{tab:comparison}
    \centering
    \resizebox{\columnwidth}{!}{
    \begin{tabular}{c|c|c|c|c|c}
    \toprule
    Property
    & \makecell{
        PIR-with-keys$^\dagger$ \\ 
        \cite{menonSPIRALFastHighRate2022, angelPIRCompressedQueries2018a, aliCommunicationComputationTradeoffs2021a}}
    & \makecell{
        Stateful PIR$^*$ \\ 
        \cite{PrivateStatefulInformation}}
    & \makecell{
        LWE-based PIR \\ (SimplePIR~\cite{henzingerOneServerPrice2023,davidsonFrodoPIRSimpleScalable2023})}
    & \makecell{Sublinear PIR \\ 
              (Piano \cite{piano})}
    & \textbf{\protocol{}}$^\dagger$ \\
    \midrule
    Client-side Storage & \cellcolor{MyGreen} - &  1 MB & \cellcolor{MyRed} 128 MB & \cellcolor{MyRed} 66 MB & \cellcolor{MyGreen} - \\
    Offline Communication & 600 KB - 13 MB & \cellcolor{MyRed} 1 GB & \cellcolor{MyRed} 128 MB & \cellcolor{MyRed} 1 GB & \cellcolor{MyGreen} 400 B \\
    Throughput & \cellcolor{MyRed} < 1 GB/s & \cellcolor{MyRed} < 500 MB/s & \cellcolor{MyGreen} 10 GB/s & \cellcolor{MyGreen} 126 GB/s & \cellcolor{MyGreen} 3 GB/s \\
    \bottomrule
    \end{tabular}
    }
\end{table}


Private Information Retrieval (PIR) enables a client to privately retrieve an element from a database without revealing to the server which element was accessed.
PIR is an essential functionality for privacy-enhancing applications such as compromised credential checks~\cite{liProtocolsCheckingCompromised2019a}, private contact discovery~\cite{Thomas2019ProtectingAF}, metadata-private applications~\cite{ahmadAddraMetadataprivateVoice2021}, private blocklist lookups~\cite{koganPrivateBlocklistLookupsa}, private STC auditing~\cite{henzingerOneServerPrice2023}, private messaging~\cite{199325}, and private distributed file systems~\cite{mazmudar2024peer2pirprivatequeriesipfs}.

Existing PIR protocols operate in the multi-server or single-server model.
Multi-server PIR, while being an efficient solution, relies on a strong assumption that the servers hold identical copies of the database and do not collude.
In contrast, single-server PIR~\cite{chorPrivateInformationRetrieval1998a}, where a single server holds the database, is compelling for the applications mentioned above since ensuring multiple servers do not collude can be challenging in practice.
Moreover, maintaining a consistent view of a dynamic database among servers requires additional coordination.
Hence, in this paper, we focus our attention on single-server PIR schemes.

Clients in the applications mentioned above typically have limited storage and computational resources, e.g., browsers and smartphones.
Hence, we can classify existing protocols based on the burden they place on the client.


\textbf{PIR using client-specific keys.}
This approach to PIR originated from XPIR~\cite{aguilar-melchorXPIRPrivateInformation2016} and was followed by recent works~\cite{angelPIRCompressedQueries2018a, aliCommunicationComputationTradeoffs2021a, mugheesOnionPIRResponseEfficient2021, ahmadAddraMetadataprivateVoice2021, menonSPIRALFastHighRate2022}.
While these protocols require no client-side storage, they still suffer from low throughput, up to hundreds of megabytes per second~\cite{menonSPIRALFastHighRate2022}.
The primary reason for this low throughput is that these protocols require a large number of public-key operations, usually linear in the database size.

\textbf{PIR with client storage.}
Private stateful information retrieval~\cite{patelPrivateStatefulInformation2018} reduces server computation but requires clients to maintain a large, database-dependent private state, which is difficult for a resource-constrained client.
Additionally, the hint must be updated when the database is updated.
More recent high-throughput protocols based on Learning With Errors (LWE), such as SimplePIR~\cite{henzingerOneServerPrice2023}, FrodoPIR~\cite{davidsonFrodoPIRSimpleScalable2023}, and DoublePIR~\cite{henzingerOneServerPrice2023}, achieve high throughput that approaches memory bandwidth, requiring only simple matrix-vector multiplication in the online phase.
However, they require clients to store large hints (128 MB for SimplePIR) that must be updated whenever the database changes, making them impractical for real-world scenarios.
Another paradigm takes throughput optimization to the extreme by proposing protocols with sublinear online time~\cite{piano, corrigan-gibbsSingleServerPrivateInformation2022, thorpir, 10.1145/3658644.3690266}.
Such protocols perform offline preprocessing such that the online work is only sublinear in the database size, which is a substantial improvement for very large databases.
However, these protocols impose impractical client-side requirements.
For instance, Piano~\cite{piano} requires streaming the entire database to the client during preprocessing, while other approaches~\cite{corrigan-gibbsSingleServerPrivateInformation2022, thorpir, 10.1145/3658644.3690266} demand substantial client-side storage and computation, contradicting the resource constraints of limited client devices.

\Cref{tab:comparison} summarizes the limitations of existing work.
Given these limitations in existing approaches, we pose the following question:
\begin{center}
    \textit{Can we construct a high-throughput PIR protocol without placing storage burden on the client?}
\end{center}

In this paper, we answer this question affirmatively by introducing \protocol{}, a new PIR protocol that achieves both high throughput and practical client-side requirements. 
Compared to existing PIR protocols that fulfill the mentioned requirements such as HintlessPIR and YPIR, we show that \protocol{} is up to 2.2x faster.
Moreover, \protocol{} is the first PIR protocol using Paillier that achieves throughput that is competitive with state-of-the-art PIR protocols.

\subsection{Technical Overview}

At the core of our approach is a novel technique for compressing homomorphic ciphertexts based on LWE and Ring-LWE. While our approach is independently applicable to protocols using homomorphic encryption, we apply it to address the challenges of PIR.

\textbf{Compressing LWE Ciphertexts.}
(R)LWE ciphertexts are commonly sent over the network in modern single-server PIR protocols~\cite{henzingerOneServerPrice2023, menonSPIRALFastHighRate2022, menonYPIRHighThroughputSingleServer2024, angelPIRCompressedQueries2018a} and other applications that use homomorphic encryption~\cite{akhavanmahdaviLevelPrivateNonInteractive2023, chenLabeledPSIFully, chenLabeledPSIFully2018, mahdaviPEPSIPracticallyEfficient2023}.
However, they suffer from high ciphertext expansion factors, which is the ratio of ciphertext size to its corresponding plaintext size.
This expansion factor directly impacts communication costs in client-server architectures such as PIR~\cite{regevLatticesLearningErrors2009, lyubashevskyIdealLatticesLearning2010}, through both query uploads (client-generated ciphertexts) and response downloads (ciphertexts processed by the server).
While there are effective methods to reduce upload costs~\cite{aliCommunicationComputationTradeoffs2021a,choTranscipheringFrameworkApproximate2021a,dobraunigPastaCaseHybrid2021,dobraunigRastaCipherLow2018,albrechtCiphersMPCFHE2015,canteautStreamCiphersPractical2018}, the same techniques can not be applied to reduce download costs.
Previous efforts to address ciphertext size include switching to smaller parameters in LWE-based schemes~\cite{naehrigCanHomomorphicEncryption2011a}, performing modulus switching in RLWE-based schemes~\cite{MicrosoftSEALRelease2023, cheonHomomorphicEncryptionArithmetic2017a}, and LWE-to-RLWE packing~\cite{chenEfficientHomomorphicConversion2021}. 
However, an LWE ciphertext size remains proportional to $O(n\log q)$, where $n$ and $q$ represent the ciphertext dimension and modulus, respectively.

Our proposed compression technique reduces the size of and (R)LWE ciphertext by exploiting the linear step in the decryption of (R)LWE ciphertexts, which reduces a ciphertext along its dimension, $n$. 
Since this step is linear, it can be performed homomorphically on the server side using an additive homomorphic encryption scheme (e.g., Paillier~\cite{paillierPublicKeyCryptosystemsBased1999}), allowing the server to return a single additive ciphertext to the client.
This compression technique achieves size reduction of approximately 89\% (from 6.8KB to 768 B) for a single LWE ciphertext and up to 99\% for a batch of LWE ciphertexts.
The only additional requirement is a small, reusable compression key — the encrypted (R)LWE secret under the additive scheme.

\textbf{Constructing \protocol{}.}
Utilizing this compression technique, \protocol{} addresses the main bottleneck of the high-throughput approach of SimplePIR~\cite{henzingerOneServerPrice2023}: the large client-side hint.
However, a direct application of compression would result in prohibitive computational costs due to expensive Paillier operations.
PIR using additive schemes such as Paillier have been known to be very slow due to these expensive operations~\cite{PrivateStatefulInformation}.

\protocol{} employs an offline/online paradigm to push all expensive Paillier operations to the offline phase that can occur during the server's idle time.
We use an alternative representation for Paillier ciphertexts which enables an offline/online split, whilst maintaining security.
Importantly, under computational assumptions, the client does not need to communicate with the server during the offline phase, except for a one-time Paillier public key and cryptographic seed.
This leaves the online phase with only two matrix-vector multiplications, which are sufficient to achieve throughput comparable to SimplePIR.
Since the second matrix multiplication is over a large modulus, we represent the operands in RNS form to enable the use of efficient machine-level operations.
Importantly, upon any database updates, the server can redo the precomputation without any communication with the client.
While this requires the server to store a per-client state, the nature of this state in \protocol{} differs fundamentally from related work and is novel in the literature.
Other protocols store multiple megabytes of static, per-client information that does not require updates over time.
In contrast, \protocol{} maintains a much smaller per-client state—only a few hundred kilobytes—which is refreshed per query, though this update can also be performed silently, i.e., without interaction with the client, and during the server's idle time.
While using an offline phase is very common for PIR protocols using lattice-based encryption schemes, our work is the first to propose a practical offline/online protocol using Paillier.

\textbf{Summary of Results.}
Our results demonstrate that \protocol{} achieves a throughput of 3 GB/s for a 1GB database and 3.5 GB/s for a 2GB database, while maintaining only a 120 and 170 KB per-client per-query state, respectively.
\protocol{}'s throughput is up to 10x higher than that of protocols using client-specific keys.
Aside from competitive runtime, our work is the first to demonstrate that PIR using additive schemes such as Paillier can be efficient and competitive with pure lattice-based schemes.


%% file: background.tex
\section{Background}
\label{sec:background}

In this section and throughout the paper, we index the i$^{th}$ element of the vector $\textbf{a}$ as $\textbf{a}[i]$.
We also define $[n]=\{0,1,\cdots,n-1\}$ and $\lfloor\cdot\rceil$ denotes rounding to the nearest integer.
$x\leftarrow D$ denotes the variable $x$ sampled from a distribution $D$ and $x\sample S$ denotes sampling $x$ uniformly from a set $S$.

\subsection{Homomorphic Encryption}
Homomorphic Encryption (HE) is a form of public-key cryptography which permits computation on messages while in encrypted form, without the need to access the secret key. Similar to other public-key cryptosystems, homomorphic ciphertexts are larger than the underlying plaintext.
The ratio between the ciphertext and plaintext is denoted as the \textit{expansion factor}.


\subsubsection{LWE \& RLWE ciphertexts}
\label{sec:lwe}

For this work, we describe a simple version of an encryption system based on the Learning With Errors (LWE)~\cite{regevLatticesLearningErrors2009} assumption which we will denote by $\mathcal{E}_{\text{LWE}}$. The most prominent encryption schemes that have ciphertexts of this format are Regev~\cite{regevLatticesLearningErrors2009}, FHEW~\cite{ducasFHEWBootstrappingHomomorphic2015}, and TFHE(CGGI)~\cite{chillottiTFHEFastFully2020}.

$\mathcal{E}_{\text{LWE}}$ uses the following parameters:
dimension $n$, ciphertext modulus $q$, plaintext modulus $p$, noise scale $\Delta=\round{q/p}$, a discrete error distribution over $\ZZ_q$ called $\chi$. We sample the secret key, $\texttt{sk}$, from $\ZZ_q^{n}$. The encryption and decryption procedure for $\mathcal{E}_{LWE}$ is shown in \Cref{alg:lwe-encrypt-decrypt}.

\begin{algorithm}[H]
     \caption{Encryption and Decryption of $\mathcal{E}_{LWE}$}
     \label{alg:lwe-encrypt-decrypt}
     \begin{algorithmic}[1]
        \vspace{1mm}
        \Procedure{LWEEncrypt}{$\texttt{sk}, \mu$}
        \State Sample $\textbf{a}\xleftarrow{\$} \ZZ_q^{n}$ and $e \leftarrow \chi$
        \vspace{1mm}
        \State $b = \sum_{i\in[n]} \textbf{a}[i] \cdot \texttt{sk}[i] + \Delta \cdot \mu + e \mod q$
        \vspace{1mm}
        \State \Return $\ct=(\textbf{a},b)$
        \EndProcedure
        \vspace{2mm}
        \Procedure{LWEDecrypt}{$\texttt{sk}, \ct=(\textbf{a},b)$}
        \State $\mu^* = \left(b - \sum_{i\in[n]} \textbf{a}[i] \cdot \texttt{sk}[i]\right) \mod q $
        \vspace{1mm}
        \State $\mu'=\lfloor\mu^* / \Delta\rceil$
        \vspace{1mm}
        \State \Return $\mu'$
        \EndProcedure

        

        
        
     \end{algorithmic}
\end{algorithm}

Fresh ciphertexts can be compressed to reduce network costs. Since $\textbf{a}$ is sampled at random, we can send the seed used to generate $\textbf{a}$ instead of the vector itself. Concretely, instead of sending $c=(\textbf{a},b)$, the client can produce $\bar{c}=(\theta,b)$ where $\theta\leftarrow\{0,1\}^{\lambda}$ is the seed of a cryptographically secure PRG used to generate $\textbf{a}$, i.e., $\textbf{a}\leftarrow \texttt{PRG}(\theta)$. With this technique, fresh ciphertexts are only $\lambda+\log_2 q$ bits instead of $n\log_2 q$.

Similar to LWE, we can also construct an encryption scheme based on the Ring Learning with Errors (RLWE)~\cite{lyubashevskyIdealLatticesLearning2010} assumption.
Cryptosystems such as BGV~\cite{brakerskiLeveledFullyHomomorphic2012}, BFV~\cite{brakerskiFullyHomomorphicEncryption2012,fan2012somewhat}, and CKKS~\cite{cheonHomomorphicEncryptionArithmetic2017a} have ciphertexts of a similar format.
More details regarding RLWE cryptosystems are provided in \Cref{sec:rlwe-compressing}.


\subsubsection{Paillier Cryptosystem}
The Paillier cryptosystem~\cite{paillierPublicKeyCryptosystemsBased1999} is an additive cryptosystem based on the computational intractability of the quadratic residuosity problem.
The plaintext space of Paillier is $\ZZ_m$ for $m=pq$, for two secret primes $p$ and $q$ and a Paillier ciphertext is an element $c\in\ZZ_{m^2}$.
The cryptosystem supports homomorphic additions between two ciphertexts, and plaintext multiplications.
Throughout this paper, we denote Paillier addition and plaintext multiplication as $\oplus$ and $\otimes$, respectively.
We also define $\textsc{PaillierMatMul}_m(C,P)$ which denotes a matrix multiplication between a matrix of ciphertexts $X\in\ZZ_{m^2}^{a\times b}$ and matrix of plaintext values $P\in\ZZ_m^{b\times c}$.

\subsection{Private Information Retrieval}
Private Information Retrieval (PIR) is a protocol where a client retrieves an element from a database such that the query is not revealed.
A specific variant of PIR dubbed \textit{Offline/Online PIR} aims to push the costly operations to the offline phase, such that the online phase is very efficient.
We rely on the definition proposed by Corrigan-Gibbs and Kogan~\cite{corrigan-gibbsSingleServerPrivateInformation2022}, where an offline/online PIR scheme is a tuple $\Pi = (\textsc{Setup}, \textsc{Hint}, \textsc{Query}, \textsc{Answer}, \textsc{Reconstruct})$ of five algorithms:

\begin{itemize}\itemsep0mm
    \item $(\ck, q_h) \leftarrow \textsc{Setup}(1^\lambda, N)$ : a randomized algorithm that outputs a client key ($\ck$) and hint request ($q_h$), given the security parameter and database size.
    \item $h \leftarrow \textsc{Hint}(\db, q_h)$ : a deterministic algorithm that outputs a client-specific hint, given the database and the hint request.
    \item $\qu \leftarrow \textsc{Query}(\ck, i)$ : a randomized algorithm that generates a query, given the client key and a desired index.
    \item $\ans \leftarrow \textsc{Response}(\db, \qu)$ : a deterministic algorithm that produces an answer, given the query and the database.
    \item $d \leftarrow \textsc{Extract}(h, \ans)$ : a deterministic algorithm that outputs a final response, given the hint and the answer.
\end{itemize}


\begin{definition}[Correctness]
    A PIR protocol with preprocessing consisting of the five aforementioned routines is $\delta$-correct, if for a domain $\mathcal{D}$, any database $\db\in \mathcal{D}^{N}$ and any $i\in[N]$,
    \begin{align}
    \PP\left[\db[i] = f \middle|
        \begin{array}{r c l} 
            (\ck, q_h) & \leftarrow & \textsc{Setup}(1^\lambda, N) \\
            h & \leftarrow & \textsc{Hint}(\db, q_h) \\
            \qu & \leftarrow & \textsc{Query}(\ck, i) \\
            \ans & \leftarrow & \textsc{Response}(\db, \qu) \\
            f & \leftarrow & \textsc{Extract}(h, \ans)
        \end{array}
    \right] > 1-\delta
    \end{align}
\end{definition}

Intuitively, the two pieces of information that the server observes, $q_h$ and $\qu$, should not reveal any information about the client's query.
This is formalized in the following definition for single-server security of offline/online PIR~\cite[Definition 46]{10.1007/978-3-030-45721-1_3}:

\begin{definition}[Security]
     A PIR protocol is $\epsilon$-secure if for any PPT adversary $\adv$ and any $i,j\in[N]$,

\scalebox{0.85}{%
\begin{minipage}{\columnwidth}    
    \begin{align*}
        \left| \PP\left[\adv(1^N, (q_h, \qu)) = 1 \, \middle| \,
            \begin{array}{c}
                (\ck, q_h) \leftarrow \textsc{Setup}(1^\lambda, N) \\
                \qu \leftarrow \textsc{Query}(\ck, i)
            \end{array}
        \right] \right.\\
        \left. {} - \PP\left[\adv(1^N, (q_h, \qu)) = 1 \, \middle| \,
            \begin{array}{c}
                (\ck, q_h) \leftarrow \textsc{Setup}(1^\lambda, N) \\
                \qu \leftarrow \textsc{Query}(\ck, j)
            \end{array}
        \right] \right| \leq \epsilon
    \end{align*}
\end{minipage}
}

\end{definition}

%% file: constructions.tex
\section{Additive HE for Smaller FHE Responses}
\label{sec:main}

We propose a technique to compress LWE/RLWE ciphertexts using auxiliary information provided by the client.

\textbf{Exploiting Linear Phase Evaluation.}
In LWE and RLWE decryption, we compute an intermediate value which is commonly referred to as the \textit{phase}, i.e., $\mu^*$ in \Cref{alg:lwe-encrypt-decrypt}.
Phase evaluation is linear and the phase is much smaller than the ciphertext itself.
The main insight behind our solution is to homomorphically compute the phase on the server using encrypted values of the secret key, encrypted under an additive encryption scheme.
Since the phase is much smaller than the original ciphertext, this results in a smaller response size.
In general, our technique can be applied to any encryption scheme that has a linear phase evaluation. Examples of encryption schemes with this property are Regev~\cite{regevLatticesLearningErrors2009}, FHEW~\cite{ducasFHEWBootstrappingHomomorphic2015}, TFHE~\cite{chillottiTFHEFastFully2020}, BFV~\cite{fan2012somewhat,brakerskiFullyHomomorphicEncryption2012}, and BGV~\cite{brakerskiLeveledFullyHomomorphic2012}.

In the following section, we describe the procedure for compressing LWE ciphertexts and follow with additional optimizations to this procedure.
In \Cref{sec:rlwe-compression}, we describe compression for RLWE-based ciphertexts as well.

\textbf{The Additive Encryption Scheme.}
For the compression protocol, we require an additive encryption scheme which we denote $\mathcal{E}_A$ such that the plaintext space is $\ZZ_m$, for some $m$.
Also, denote the ciphertext space of $\mathcal{E}_A$ as $\mathcal{C}$.
$\mathcal{E}_A$ supports addition and plaintext multiplications.
We denote addition and plaintext multiplication with $\oplus$ and $\otimes$, respectively.
Moreover, denote the secret key generated by $\mathcal{E}_A$ as $\addkey$ and the corresponding encryption and decryption algorithms as $\texttt{AEnc}$ and $\texttt{ADec}$.

Paillier~\cite{paillierPublicKeyCryptosystemsBased1999}, Damgard-Jurik~\cite{damgardGeneralisationSimplificationApplications2001a}, Exponential ElGamal~\cite{elgamalPublicKeyCryptosystem1985}, and Benaloh~\cite{benalohDenseProbabilisticEncryption1994} are examples of cryptosystems that can be used for this purpose.

\subsection{Compressing LWE Ciphertexts}

The ciphertext compression algorithm for LWE and the corresponding modified decryption algorithm are given in \Cref{alg:lwe-compress}.

\begin{algorithm}[H]
    \caption{LWE compression, performed by the server and the corresponding modified decryption process, performed by the client over a compressed ciphertext. The compression key $\ck\in\mathcal{C}^{n}$ is such that $\ck[i]=\texttt{AEnc}(\addkey, \sk[i])$.}
    \label{alg:lwe-compress}
    \begin{algorithmic}[1]
        \Procedure{LWECompress$_{q}$}{$\ck, \ct=(\textbf{a},b)$}
        \Comment{$\ct\in\ZZ^{n}\times\ZZ$}
        \State $x=b$
        \For{$i \in [n]$} 
            \State $x \leftarrow x \oplus \left( (q-\textbf{a}[i]) \otimes \ck[i] \right)$\label{alg:line-multiply}
        \EndFor
        \State \Return $x$ \Comment{$\mu^*=\texttt{ADec}(\addkey, x)$}
        \EndProcedure
    \vspace{3mm}
    \Procedure{ModifiedLWEDecrypt$_{q, p}$}{$\addkey,x$}
	 	\State $\mu^{**} = \texttt{ADec}(\addkey, x) \mod q$
            \label{alg:lwe-mod-decrypt}
	 	\vspace{1mm}
	 	\State $ \mu'' = \lfloor \mu^{**}/\Delta\rceil$
            \Comment{$\Delta=\round{q/p}$}
            \vspace{2mm}
            \State \Return $\mu'' \in \ZZ_{p}$
   \EndProcedure
   \end{algorithmic}
\end{algorithm}

\begin{theorem}[Correctness]
\label{thm:lwe-compress-correct}
    For an LWE ciphertext $\ct\in\ZZ_q^{n+1}$, if $m>q+nq^2$, then $\textsc{LWECompress}_q$ produces a compressed ciphertext which decrypts to the correct message if decrypted using \textsc{ModifiedLWEDecrypt}. More formally, if 

\begin{align*}
    x\leftarrow\textsc{LWECompress}_{q}(\ck, \ct) \\
    \mu'' \leftarrow \textsc{ModifiedLWEDecrypt}_{q,p}(\addkey, x)
\end{align*}
then
$\mu'' = \textsc{LWEDecrypt}(\sk, \ct)$
\end{theorem}

\begin{proof}
In the \textsc{ModifiedLWEDecrypt}$_{q,p}$ procedure (\Cref{alg:lwe-mod-decrypt} of \Cref{alg:lwe-compress}), we calculate 
$b + \sum_{i\in[n]} (q-\textbf{a}[i]) \cdot \sk[i]$, encrypted under additive encryption, which is achievable due to the linear properties of the additive encryption. We know that $\sk[i], \textbf{a}[i]$ and $b$ are elements in $\ZZ_q$ so $0 \leq \sk[i], \textbf{a}[i], b < q$ and 

{\footnotesize
\begin{align}
    b + \sum_{i\in[n]} (q-\textbf{a}[i]) \cdot \sk[i] \leq q + \sum_{i\in[n]} q  \cdot q = q + nq^2 < m .
\end{align}
}

so there is no overflow in the plaintext space of the additive ciphertext.
In $\textsc{ModifiedLWEDecrypt}_{q,p}$ (\Cref{alg:lwe-compress}), we have

{\footnotesize
\begin{align*}
    \mu^{**}\mod q &= \texttt{ADec}(\addkey, x) \mod q \\
    &= \left((b + \sum_{i\in[n]} (q-\textbf{a}[i]) \cdot \sk[i]) \mod m \right) \mod q \\
    &= \left(b + \sum_{i\in[n]} (q-\textbf{a}[i]) \cdot \sk[i] \right) \mod q \\
    &= b - \sum_{i\in[n]} \textbf{a}[i] \cdot \sk[i] \mod q
\end{align*}
}

This is identical to $\mu^*$ in line 1 of \Cref{alg:lwe-encrypt-decrypt}, hence, since the subsequent steps of \textsc{LWEDecrypt} and \textsc{ModifiedLWEDecrypt} are identical, they produce the same response, and the theorem is proven.
\end{proof}

In cryptosystems such as TFHE~\cite{chillottiTFHEFastFully2020}, the secret key is sampled from a binary distribution.
In such a case, we can tighten the inequality required in \Cref{thm:lwe-compress-correct} for correctness because $0\leq\sk[i]\leq 1$. The following corollary summarizes this fact.

\begin{corollary}
    If the LWE secret key is binary and $m>q+nq$, \textsc{LWECompress} produces a compressed ciphertext which decrypts to the correct message if decrypted using \textsc{ModifiedLWEDecrypt}.
\end{corollary}

\textbf{Security.}
In Gentry's original construction of a bounded depth encryption scheme, he proposed the idea of using a chain of semantically secure cryptosystems, such that each cryptosystem encrypts the secret key of the next~\cite{gentryFullyHomomorphicEncryption2009}. Gentry proved that if the secret key of each cryptosystem is sampled independently, the composed scheme is also semantically secure.

Let $\mathcal{E}'$ denote the cryptosystem which is the chaining of $\mathcal{E}_{\text{LWE}}$ and $\mathcal{E}_{A}$. The encryption and decryption procedure of $\mathcal{E}'$ is shown in \Cref{alg:lwe-encrypt-decrypt} and \Cref{alg:lwe-compress}, respectively. The secret key of $\mathcal{E}'$ is the combination of the secret keys of $\mathcal{E}_{\text{LWE}}$ and $\mathcal{E}_{A}$. The same holds for the public key as well. Moreover, we also release encryptions of the bits of the secret key of $\mathcal{E}_{\text{LWE}}$ under the secret key of $\mathcal{E}_{A}$. 

\begin{proposition}[Security]
    If $\mathcal{E}_{\text{LWE}}$ and $\mathcal{E}_A$ are semantically secure, then $\mathcal{E}'$ is also semantically secure.
\end{proposition}

\subsection{Using Smaller Compression Keys}
\label{sec:smaller-compression-key}
In practice, the plaintext space of the additive encryption system could be much larger than is required for the correctness of the compression technique to hold, i.e., $m \gg q + n q^2$.
For example, the plaintext space of Paillier for 128-bit security is 3072 bits, which is much larger than $q + n q^2$ for any common choice of LWE parameters.
We can use this gap to pack multiple digits of the LWE secret key within one additive ciphertext.
Instead of encrypting each bit of the LWE key separately, we encrypt the first $t$ digits of the secret key together into one packed additive ciphertext as $\texttt{pck}_{0-t} = \texttt{AEnc}(\addkey, \sum_{i\in[t]}\sk[i] \cdot \delta^{i})$ for a large enough $\delta$.
Specifically, $\delta$ should be such that $\delta > q + n q^2$ (or $\delta > q + n q$ in the case of binary keys).
On the server side, the server unpacks the secret key by computing $\ck[i] = \delta^{t-1-i} \otimes \texttt{pck}_{0-t}$ for $i\in[t]$.
Compression proceeds as before, with the only difference being that the encrypted phase, calculated by the server in the additive ciphertext, is scaled by a factor of $\delta^{t-1}$.
\appsection{sec:lwe-compress-packed-keys} details the procedures for generating the packed key, unpacking it, and the corresponding modified LWE decryption function.
We use the same function for compressing the ciphertext.


\subsection{Batched Compression}
\label{sec:batched-compression}

To achieve better compression, multiple LWE ciphertexts (encrypted using the same secret key) can be compressed within the same additive ciphertext, which we denote as \textit{batched compression}.
Each LWE ciphertext takes up $\log_2 (q+nq^2)$ bits of the total bitwidth of the plaintext space.
So, if $m$ is the modulus of the plaintext space, then $\floor{\log_2 m / \log_2 (q+nq^2)}$ LWE ciphertexts can be compressed into one ciphertext from the additive cryptosystem.

\Cref{alg:batched-compress} illustrates how to compress $\ell$ LWE ciphertexts within one additive ciphertext. The corresponding decryption procedure is also shown.
Using \textsc{LWECompress} as a subprocedure allows for better parallelization when compressing many LWE ciphertexts.

\begin{algorithm}[H]
	 \caption{Batch Compression of LWE ciphertexts by the server and the modified decryption procedure, performed by the client. The compression key $\ck$ is such that $\sk[i]={\texttt{ADec}(\addkey, \ck[i])}$ and $cts=\{c_j\}_{j\in [\ell]}$ such that $c_j = (\boldsymbol{\textbf{a}_j}, b_j)\in \ZZ_{q}^{n}\times\ZZ$ and $\gamma = q + n q^2$.}
	 \label{alg:batched-compress}
	 \begin{algorithmic}[1]
    \Procedure{BatchedLWECompress$_{q, \gamma}$}{$\ck, cts=\{c_j\}_{j\in [\ell]}$}
        \For{$j\in{[\ell]}$}
            \State $x_j \leftarrow \textsc{LWECompress}_q(\ck,c_j)$
            \State $x \leftarrow x \oplus \gamma^{j} x_j $
        \EndFor
        \Return $x$
    \EndProcedure
    \vspace{3mm}
    \Procedure{ModifiedBatchedLWEDecrypt$_{q,p, \gamma}$}{$\addkey, x$}
	 	\State $\mu^{**} = \texttt{ADec}(\addkey, x)$
            \For{$j \in [\ell]$}
                \State $\mu_j^{**} =\floor{\mu^{**}/\gamma^j} \mod \gamma$
	       \State $\mu''_j = \lfloor \frac{\mu_j^{**} \mod q}{\Delta} \rceil$\Comment{$\Delta=\round{q/p}$}
            \EndFor
        \Return $\{\mu_j''\in \ZZ_{p}\}_{j\in[\ell]}$
    \EndProcedure
	 \end{algorithmic}
\end{algorithm}

\begin{theorem}[Correctness]
    Let $\ct=\{c_j\}_{j\in [\ell]}$ be a vector of $\ell$ LWE ciphertexts.
    For $\gamma\geq q + n q^2$, if  $m > \gamma^{\ell}$, then \textsc{BatchedLWECompress}$_{q, \gamma}$ produces a compressed ciphertext which, if decrypted using the corresponding modified decryption, decrypts to the vector of $\ell$ plaintexts. More formally, if
    \begin{align*}
        x \leftarrow \textsc{BatchedLWECompress}(\ck, \ct, k) \\
        \{\mu'_j\}_{j\in\ell} \leftarrow \textsc{ModifiedBatchedLWEDecrypt}(
        \addkey, x)
    \end{align*}
    then $\mu'_j = \text{LWEDecrypt}(\sk,c_j)$.
\end{theorem}
\begin{proof}
    By the proof of \Cref{thm:lwe-compress-correct} we know that if $\mu^{**}_j=\texttt{ADec}_{\texttt{s}}(x_j)$, then $0\leq \mu^{**}_j < \gamma = q + nq^2$. Hence, we have
    \begin{align*}
        \mu^{**} = \sum_{j\in[\ell]} \gamma^j \mu^{**}_j \leq \sum_{j\in[\ell]} \gamma^j (\gamma-1) = \gamma^{\ell} - 1 < \gamma^{\ell} < m .
    \end{align*}
    Hence, the plaintext corresponding to $x$, i.e., $\mu^{**}$, does not overflow in the plaintext space of the additive ciphertext.
    If $\mu^*_j$ is equivalent to $\mu^*$ in the \textsc{LWEEncrypt} procedure, then for some value $t$,
    \begin{align*}
        \mu''_j =\floor{\mu^{**}/\gamma^j} \mod \gamma = (\mu^*_j + \gamma \cdot t ) \mod \gamma = \mu^*_j
    \end{align*}
    and the subsequent steps are similar, which proves the theorem.
\end{proof}

\subsubsection{Faster Batched Compression with Expanded Key}
Compression makes use of expensive operations in the additive scheme.
The plaintext multiplication in \Cref{alg:line-multiply} of \Cref{alg:lwe-compress} is the most expensive operation.
In additive schemes such as Paillier and ElGamal, this is equivalent to a modular exponentiation in a large group.

In the batched setting, we can reduce the overhead by precomputing and reusing multiples of the bits of the secret key. If we decompose $(q-\textbf{a}[i])$ as $(b_{t-1}\cdots b_1 b_0)_2 = (q-\textbf{a}[i]) \mod q$ we compute the plaintext multiplication as follows
\begin{align}
    &(q-\textbf{a}[i]) \otimes \ck[i] \\
    &= 2^{t-1} b_{t-1} \ck[i] + \cdots + 2b_{1} \ck[i] + b_{0} \ck[i]
\end{align}
and we can precompute an \textit{extended compression key}, $\eck$, such that $\eck[i][j] = 2^j \ck[i]$ for $j \in [t]$, which can be reused for all LWE ciphertexts we want to compress.
\Cref{alg:batched-compress-precompute} shows the procedure extending the compression key and compressing LWE ciphertexts using the extended compression key.

\begin{algorithm}[t]
	 \caption{Batch compression of LWE ciphertexts using precomputed powers. The compression key $\ck$ is such that $\sk[i]={\texttt{ADec}(\addkey, \ck[i])}$ and $cts=\{c_j\}_{j\in [\ell]}$ such that $c_j = (\boldsymbol{\textbf{a}_j}, b_j)\in \ZZ_{q}^{n}\times\ZZ$.}
	 \label{alg:batched-compress-precompute}
	 \begin{algorithmic}[1]
    \Procedure{ExpandCompressionKey$_{q}$}{$\ck$}
        \State $\eck[0] = \ck$
        \For {$i \in [t-1]$} \Comment{$t = \ceil{\log_2 q}$}
            \For {$j \in [n]$}
                \State $\eck[i+1][j] = \eck[i][j] \oplus \eck[i][j]$
            \EndFor
        \EndFor
        \Return $\eck$
    \EndProcedure
    \vspace{3mm}
    \Procedure{FastLWECompress$_q$}{$\eck, \ct=(\textbf{a},b)$}
    \State $x=b$
    \For{$i \in [n]$} 
        \State $(b_{t-1}\cdots b_1 b_0)_2 \leftarrow (q-\textbf{a}[i]) \mod q$
        \Comment{$t = \ceil{\log_2 q}$}
        \For {$j \in [t]$}
            \If {$b_j = 1$}
                \State $x \leftarrow x \oplus \eck[j][i]$
            \EndIf
        \EndFor
    \EndFor
    \Return $x$ \Comment{$\mu^*=\texttt{ADec}(\addkey, x)$}
    \EndProcedure
    \vspace{3mm}
    \Procedure{FastBatchedLWECompress$_{q,\gamma}$}{$\ck, cts=\{c_j\}_{j\in [\ell]}$}
        \State $\eck \leftarrow \textsc{ExpandCompressionKey}_{q}(\ck)$
        \State $\gamma = q + n q^2$
        \For{$j\in{[\ell]}$}
            \State $x_j \leftarrow \textsc{FastLWECompress}_q(\eck,c_j)$
            \State $x \leftarrow x \oplus \gamma^{j} x_j $
        \EndFor
        \Return $x$
    \EndProcedure
	 \end{algorithmic}
\end{algorithm}

\begin{corollary}
    Let $\ct=\{c_j\}_{j\in [\ell]}$ be a vector of $\ell$ LWE ciphertexts.
    For $\gamma\geq q + n q^2$, if  $m > \gamma^{\ell}$, if
    \begin{align*}
        x \leftarrow \textsc{FastBatchedLWECompress}(\eck, \ct, k) \\
        \{\mu'_j\}_{j\in\ell} \leftarrow \textsc{ModifiedBatchedLWEDecrypt}(
        \addkey, x)
    \end{align*}
    then $\mu'_j = \text{LWEDecrypt}(\sk,c_j)$.
\end{corollary}

\subsubsection{Rescaling for Compression}

In some instances, it is possible to rescale the elements in the ciphertext to a smaller modulus without altering the underlying message.
This technique, also called modulus switching, is commonly used in the literature to simplify the decryption procedure or control noise growth~\cite{brakerskiFullyHomomorphicEncryption2012}.
However, rescaling is only possible if the noise of the underlying LWE ciphertext is less than a given bound. 
In \appsection{sec:modulus-switching-theorem}, we prove how rescaling is possible for LWE ciphertexts with binary keys, if the noise is less than a certain bound, i.e., less than $\Delta/4$.
Rescaling to a smaller modulus accelerates our compression technique since the scalar multiplication in the additive encryption scheme is done with a smaller scalar.

\subsubsection{Better compression with a smaller scale}
\label{sec:overlapping-noise}

The number of LWE ciphertexts that fit within each additive ciphertext is determined by the scale, i.e., $\gamma=q+nq^2$.
Using a smaller scale would allow us to pack more LWE ciphertexts within each additive ciphertext.
There are two instances where we can use a smaller scale.
First, when the LWE secret key is binary.
In that case, we can use $\gamma=q+nq$ as the scale.
This follows from the fact that in the case of binary keys, $0<\mu_j^{**} \leq \gamma = q+nq$.

The second instance where we can reduce the scale is by allowing $\mu_{j}^{**}$ and $\mu_{j+1}^{**}$ to overlap in the additive scheme.
Intuitively, this is possible because the high-order bits of $\mu_{j}^{**}$ are removed when it is taken modulo $q$ as part of the modified decryption.
The lower order bits of $\mu_{j+1}^{**}$ are also rounded during the modified decryption so it is possible to add additional error, as long as it does not interfere with the message.
Specifically, if $|e|<\Delta/4$ (instead of the usual condition where $|e|<\Delta/2$ for correct decryption), we can reduce the scale to $\gamma=q^2$ and $\gamma=q$ in the case of non-binary and binary keys, respectively.
Due to space restrictions, we provide proof of the correctness of this technique using a smaller scale under these conditions in the full version of the paper.

%% file: zippir.tex
\section{Our PIR Construction: \protocol{}}

In this section, we present \protocol{}, a high-throughput single-server PIR protocol optimized for resource-bound clients.
\protocol{} uses the compression technique from the previous section but pushes all expensive operations to the offline phase.
We show how we benefit from our proposed compression technique whilst maintaining an efficient online phase.

\subsection{Leveraging the offline phase}
As we saw in the previous section, our compression technique is very effective in reducing the size of LWE ciphertexts.
However, this comes at a significant cost due to the expensive operations in the additive scheme.
In the case of Paillier, multiplications are done via expensive exponentiations.

To avoid performing expensive operations in the online phase, we note an observation by Beck~\cite{beckRandomizedDecryptionRD2015}, which states that a Paillier encryption of a message $\mu\in\ZZ_m$ can be represented as 
\begin{align}
    \enc(\mu) = \enc(r) + (\mu-r)    
\end{align}
where $r$ is a uniformly random value over the plaintext space of the encryption scheme.
Moreover, instead of sampling $r$ and encrypting it, we can directly sample the ciphertext $\enc(r)$ from $\ZZ_m^2$ using a PRNG seed.
While valid Paillier ciphertexts are elements of $\ZZ_{m^2}^*$, the likelihood of sampling an invalid ciphertext is
negligible in the security parameter (roughly $2/\sqrt{m}$).

Using this observation, we can evaluate any linear function (such as the compression function) over Paillier ciphertexts in two steps.
First, we compute the function over the random component $\enc(r)$, which does not depend on the message $\mu$.
Then we compute the function over $(\mu-r)$, which we call the \textit{offset}, and add this to the result.
Using this approach, all expensive operations over the additive ciphertexts are performed in an offline phase, before the value of $\mu$ is known.


\subsection{\protocol{} Description}
In this section, we describe our construction, \protocol{}.
We also show how, through simple modifications, \protocol{} requires no storage overhead for the client and, under computational assumptions, the offline processing can be done silently, except for an initial public key and seed.

\protocol{} builds on SimplePIR~\cite{henzingerOneServerPrice2023} but eliminates the need for the client to store a large database-dependent hint by compressing the hint using the algorithms described in \Cref{sec:main}.
However, since the compression is a linear operation, we can leverage the online phase, as described in the previous section.
While we do not use the precise procedures described in \Cref{sec:main}, the operations in \protocol{} perform the same function in principle, albeit in an offline/online manner.

\Cref{alg:zippir} shows a simple description of \protocol{}, and we describe additional modifications and enhancements in the following sections.
At a high level, the protocol works as follows. In the offline phase:
\begin{itemize}\itemsep0mm
    \item At setup ($\textsc{Setup}$), the client generates the necessary Paillier public keys and samples the Paillier ciphertexts, which are the random components of the compression keys.
    \item In the offline hint generation ($\textsc{Hint}$), the server generates a compressed version of the Paillier hint given the client's Paillier public keys and the sampled ciphertexts. To avoid client-side storage, the server holds this hint until the client issues a query.
\end{itemize}

Once the client's query is known, the online phase occurs as follows:
\begin{itemize}\itemsep0mm
    \item The client's query consists of two components: the offset of the compression key and the LWE query (from SimplePIR).
    \item The server performs two matrix multiplications — one over the database and another over the hint. Both matrix multiplications are done over modular integers. The server sends the response and the associated hint to the client.
    \item The client extracts the answer from the server's response.
\end{itemize}

\newcommand{\clienthint}{\textbf{k}}
\newcommand{\clientstate}{\st}

\begin{algorithm}[]
\caption{
  Simple description of \protocol{}. $q$, $n$, and $p$ are the LWE ciphertext modulus, ciphertext dimension, and plaintext modulus, respectively.
  Database $\db\in\ZZ_p^{d_0\times d_1}$ where $N=d_0 d_1$. Also, assume $\textbf{A} \sample \ZZ_q^{d_0 \times n}$ and the server also computes $\hint = - \db^{T} \cdot \textbf{A} \in \ZZ_q^{d_1 \times n}$.
  $\textsc{PaillierMatMul}$ denotes a multiplication between a plaintext matrix and an encrypted vector.
}
	 \label{alg:zippir}
	 \begin{algorithmic}[1]
    \Procedure{Setup}{$1^\lambda$, $(d_0,d_1)$} \Comment{Client}
        \State Generate Paillier keys ($\paillierkey$, $\paillierpk=m$) with $\lambda$-bit security
        \State Sample $\ck_r \sample \ZZ^{n}_{m^2}$
        \State $\pt_r \leftarrow \textsc{PaillierDecrypt}(\paillierkey, \ck_r)$ \label{zippir:paillier-decrypt-random}
        \State \Return (($\paillierkey$, $\pt_r$), ($m$, $\ck_r$))
    \EndProcedure
    \vspace{2mm}
    \Procedure{Hint}{$\hint \in \ZZ_q^{d_1 \times n}$, ($m$, $\ck_{r})$} \Comment{Server}
        \State $\clienthint \leftarrow \textsc{PaillierMatMul}_{m}(\hint, \ck_r)$
        \Comment{$\clienthint\in\ZZ_{m^2}^{d_1}$} \label{zippir:clienthint-compute}
        \State \Return $\clienthint$
    \EndProcedure
    \vspace{2mm}
    \Procedure{Query}{$i_0\in [d_0]$, $\pt_r$} \Comment{Client}
        \State Sample LWE key $\lwekey\sample\{0,1\}^{n}$
        \State $\ck_{o} = (\lwekey - \pt_r)\mod m$  
        \label{zippir:compression-key-offset}
        \State $u_0 =$ selection vector for index $i_0$
        \State Sample $e\leftarrow \chi_e^{d_0}$
        \State $\qu_0 = \textbf{A}\cdot \lwekey + e + \Delta\cdot u_0 \mod q$ \Comment{$\Delta = \round{q/p}$}\label{zippir:lwe-encrypt}
        \State \Return $(\ck_o, \qu_0)$ \Comment{$\ck_o \in \ZZ_m^{n}, \qu_0\in\ZZ_q^{d_0}$}
    \EndProcedure
    \vspace{2mm}
    \Procedure{Response}{$\db, \hint, \qu=(\ck_{o}, \qu_0)$} \Comment{Server}
        \State $b = \db^{T} \cdot {\qu}_0 \mod q$ \label{zippir:online-lwe}
        \State $t = b + \hint \cdot \ck_o \mod m$ \label{zippir:online-compression-offset}
        \State \Return $t$
    \EndProcedure
    \vspace{2mm}
    \Procedure{Extract}{$\paillierkey$, $(\clienthint, t)$} \Comment{Client}
        \State $\mu = (t + \textsc{PaillierDecrypt}(\paillierkey, \clienthint)) \mod m$
        \State \Return $f = \round{(\mu \mod q) p/q} \mod p$
    \EndProcedure    
  \end{algorithmic}
\end{algorithm}



The following theorems hold for the correctness and security of \protocol{}.
The full proof of correctness and security are provided in \appsection{appendix:proof}.

\begin{restatable}{theorem}{protocolCorrectness}
(\protocol{} Correctness) For LWE parameters $(n,q,\chi_e, \chi_s)$ where $\chi_s$ is a discrete Gaussian with standard deviation $\sigma$, and $\chi_s$ is a binary distribution, and plaintext modulus $p$, and failure rate $\delta$ such that 
    \begin{align}
        q/p > 2p\sigma \sqrt{2d_0\ln (2/\delta)}
        \label{eq:correctness-condition}
    \end{align}
    for random $\mask\in\ZZ_q^{d_0\times n}$, for any database $\db\in\ZZ_{p}^{d_0\times d_1}$, $\hint=-\db^{T} \cdot \mask$, Paillier modulus $m$ such that $m > q+nq$, and any query $i_0 \in[d_0]$, if 
    \begin{align*}
        ((\paillierkey, \pt_r), (m, \ck_r))  \leftarrow  \textsc{Setup}(1^{\lambda}, (d_0, d_1)) \\
        \clienthint  \leftarrow  \textsc{Hint}(\hint, (m, \ck_r)) \\ 
        (\ck_{o}, \qu_0)  \leftarrow  \textsc{Query}(i_0, \pt_r) \\
        t  \leftarrow  \textsc{Response}(\db, \hint, (\ck_{o}, \qu_0)) \\
        f  \leftarrow  \textsc{Extract}(\paillierkey, (\clienthint, t))
    \end{align*}
    then $\PP[\db[i_0] = f] > 1 - \delta - 2d_0/\sqrt{m}$.
\end{restatable}

\begin{restatable}{theorem}{protocolSecurity}
(\protocol{} Security)
    Assume we have secure LWE parameters $(n,q,\chi_e, \chi_s)$ where $\chi_s$ is a discrete Gaussian with standard deviation $\sigma$, and $\chi_s$ is a binary distribution.
    Also, assume we have secure Paillier parameters.
    Then for any PPT adversary $\adv$ and any $i,j\in[d_0]$, we have 

\begin{adjustbox}{scale=0.85}
\begin{minipage}{\columnwidth}    
    \begin{align*}
        \left| \PP\left[\adv(1^\lambda, (q_h, \qu)) = 1 \, \middle| \,
            \begin{array}{c}
                ((\_, \pt_r), q_h) \leftarrow \textsc{Setup}(1^\lambda, (d_0, d_1)) \\
                \qu \leftarrow \textsc{Query}(i, \pt_r)
            \end{array}
        \right] \right.\\
        \left. {} - \PP\left[\adv(1^\lambda, (q_h, \qu)) = 1 \, \middle| \,
            \begin{array}{c}
                ((\_, \pt_r), q_h) \leftarrow \textsc{Setup}(1^\lambda, (d_0, d_1))\\
                \qu \leftarrow \textsc{Query}(j, \pt_r)
            \end{array}
        \right] \right| < \epsilon
    \end{align*}
\end{minipage}    
\end{adjustbox}
    where $\epsilon$ is negligible in the security parameter.
\end{restatable}

\subsection{Additional Optimizations \& Tradeoffs.}
\label{sec:extensions}

The description of \protocol{} from \Cref{alg:zippir} can be modified to provide bandwidth savings and reduce computational and storage burden for the client.

\textbf{Faster matrix multiplication using RNS.}
The second matrix multiplication in the online phase $(\hint\cdot\ck_o)$ is performed over large integers, which is slow.
To be able to speed up the matrix multiplication, we represent the two operands of the multiplication in the Residue Number System (RNS) format.
The hint ($\hint$) can be converted to RNS in the offline phase so we only require the offset vector ($\ck_o$) to be converted during the online phase.
Using this technique, we can perform the second matrix multiplication over native machine operations.

\textbf{Batched compression for smaller responses.}
A more efficient instantiation of \protocol{} uses our batched compression technique instead of single compression.
This is possible because the scaling used in batched compression is also a linear operation, and hence can be split into offline and online operations.
Batched compression not only reduces the size of the response, but also results in smaller matrix-vector multiplication in \Cref{zippir:online-compression-offset}.
We can also use a small scale as mentioned in \Cref{sec:overlapping-noise}.

\textbf{Lowering offline communication costs.}
To achieve a (nearly) silent offline phase, we generate $\ck_r$ using a PRG seed.
The client and server only need to agree on the initial seed for the pseudorandom generator.
Moreover, this seed can be used to generate a large number of client hints (polynomial in the security parameter) for all future queries.
The client can also store the seed, alongside its Paillier private key, to enable queries in the future, all of which require constant-size storage.
One consequence is that the client delays decryption of $\ck_r$ (\Cref{zippir:paillier-decrypt-random}) to the \textsc{Query} step.

\textbf{Smaller responses w/ Paillier addition.}
To further reduce the response size, the online response ($t$), can be combined with the offline hint (\clienthint) instead of being sent separately, at the cost of a Paillier addition.
We can then send only one Paillier ciphertext instead of one ciphertext and one plaintext, reducing the response size by 30\%.

\textbf{Supporting multiple queries.}
The current description of \protocol{} supports one query from the client.
To extend this to multiple queries, the client can generate multiple $\ck_r$ values, for as many queries as it wishes to issue, and send them to the server.
Using a PRG, as mentioned in the previous paragraph, the offline bandwidth stays constant, even for multiple queries.
The server offline time and server per-client storage scales linearly with the number of queries that the server wants to support for that client, without additional offline communication.

\textbf{Reusability of the offline phase.}
Similar to SimplePIR, the computation of the global hint $\hint$ is reused for all queries from all clients.
The client-specific hint computation (\Cref{zippir:clienthint-compute}) is not reusable across clients because it depends on the client's public key and must be computed separately for each prospective client.

\textbf{\protocol{} with client storage.}
Although we presented \protocol{} as a client-efficient PIR protocol, it can also be modified to use client storage, similar to the work of Patel et al.~\cite{patelPrivateStatefulInformation2018}.
To achieve this, the client simply retrieves and stores the hint ($\clienthint$) in the offline phase for as many queries as it will make.
A hint is consumed for each query made by the client, similar to the client state used in PSIR~\cite{patelPrivateStatefulInformation2018}.

\subsection{Concrete Costs of \protocol{}}
\label{sec:concrete-costs}

To achieve the best concrete performance (demonstrated in \Cref{sec:pir-evaluation}), we include the first three techniques from \Cref{sec:extensions} in our final construction. 
We report the concrete costs of this augmented protocol in this section.

The concrete number of operations in each routine, along with the party that performs those operations, is listed below. Let $\gamma = \ceil{(\log_2 m - \log_2 n) / \log_2 q}$ and $d_1' = \floor{d_1 / \gamma}$, then
\begin{itemize}\itemsep0mm
    \item $\textsc{Setup}$ (Client): Paillier key generation
    \item $\textsc{Hint}$ (Server) : $d_1' n$ Paillier plaintext multiplications and additions (\Cref{zippir:clienthint-compute})
    \item $\textsc{Query}$ (Client): $d_0$ LWE encryptions, (\Cref{zippir:lwe-encrypt}), $n$ random samples in $\ZZ_m$, and $n$ Paillier decryptions (\Cref{zippir:paillier-decrypt-random})
    \item $\textsc{Response}$ (Server): $N=d_0d_1$ multiplications and additions in $\ZZ_q$ (\Cref{zippir:online-lwe}) and $d_1' n$ multiplications and additions in $\ZZ_{m}$ (\Cref{zippir:online-compression-offset}).
    \item $\textsc{Extract}$ (Client): $d_1'$ Paillier decryptions
\end{itemize}

Similarly, the concrete communication costs (in bits) of the protocol are listed below. 
\begin{itemize}\itemsep0mm
    \item Client to Server (Offline, one-time): $2\log_2 m + \lambda$
    \item Client to Server (Online, per-query) : $n\log_2 m + d_0 \log_2 q$
    \item Server to Client (Online, per-query) : $3 d_1' \log_2 m$
\end{itemize}

The server-side storage per client is $2d'_1\log_2 m$ bits for each query. The client-side storage is constant, only the Paillier private key and the seed used to generate ciphertexts, which is $\log_2 m + \lambda$ bits in total.

%% file: evaluation.tex
\section{Evaluation}
\label{sec:evaluation}

We present our evaluation in two sections.
First, evaluating the compression of LWE ciphertexts using our proposed method compared to existing approaches in the literature.
Second, we compare our PIR construction, \protocol{}, with PIR protocols in the literature.

\subsection{Evaluating Ciphertext Compression}

We implemented our compression technique as a library in C++ using GMP.
We use Paillier as the additive encryption scheme, which is extended to Damgard-Jurik when we require a larger plaintext space.
We use a 3072-bit modulus for Paillier, composed of two 1536-bit primes, which is the recommended modulus size for 128-bit security~\cite{barkerRecommendationKeyManagement2020}.
We experiment with LWE and RLWE parameters satisfying 128-bit security but our methods can be applied to other LWE and RLWE parameters without any change.

We also parallelized our implementation to minimize the latency of the compression.
Specifically, we parallelize over the LWE dimension $n$ or the number of LWE ciphertexts that are compressed, depending on whichever is larger.
Using this dynamic approach, we use existing cores on our machines even when compressing few ciphertexts.

We experiment under two scenarios: 1) compressing a single LWE ciphertext or RLWE coefficient and 2) compressing many LWE ciphertexts or multiple RLWE coefficients.
The former is useful in applications with small outputs such as private inference, whereas the latter is used for applications with large outputs such as image processing.

\subsubsection{Single Compression Evaluation}

\begin{table*}
    \centering
    \caption{Evaluation of the ciphertext compression technique for a single LWE ciphertext. Three sample parameter sets are chosen for LWE-based ciphertexts.
    The first three columns are common parameter sets used in the Concrete library~\cite{zamaConcreteTFHECompiler2022}.
    The last configuration is the STD128 configuration for CGGI in OpenFHE~\cite{albadawiOpenFHEOpenSourceFully2022}.} \resizebox{\textwidth}{!}{
    \begin{tabular}{c|c|c|c|c|c|c|c|c}
    \toprule
    \multirow{2}{*}{Parameters}
    & \multicolumn{4}{c|}{LWE $(n,\log_2 q)$} 
    & \multicolumn{4}{c}{RLWE $(N,\log_2 q)$} 
    \\
        & (630, 64)
        & (742, 64)
        & (870, 64)
        & (1305, 11)
        & (1024,27)
        & (2048,54)
        & (4096,36)
        & (8192,43)
         \\
    \midrule                     
        Compression Time   & 9.7 ms & 11.0 ms & 12.9 ms & 16.6 ms & 7.2 ms & 23.8 ms & 33.8 ms & 83.3 ms \\
        Compressed Ciphertext    & 768 B   & 768 B  & 768 B & 768 B   & 768 B  & 768 B & 768 B & 768 B \\
        Uncompressed Ciphertext  & 5.05 KB & 5.94 KB & 6.97 KB & 1.80 KB & 3.46 KB & 13.83 KB & 18.44 KB & 44.04 KB \\
        Size Reduction            & \textbf{84.78 \%} & \textbf{87.08} \% & \textbf{88.98\%} & \textbf{57.23\%} & \textbf{77.80\%} & \textbf{94.45\%} & \textbf{95.83\%} & \textbf{98.26\%} \\
    \bottomrule
    \end{tabular}
    }
    \label{tab:evaluation-lwe-compress}
\end{table*}

\Cref{tab:evaluation-lwe-compress} summarizes the results for compressing a single LWE ciphertext. We choose LWE parameters adopted in existing libraries implementing variants of LWE encryption~\cite{zamaConcreteTFHECompiler2022,albadawiOpenFHEOpenSourceFully2022}.
The results show that we consistently provide high compression rates. Notably, for $\log_2 q = 64$, our compression rates are over 84\%.

Similarly, for compression of a single RLWE coefficient, we use common parameters for RLWE-based schemes such as BFV~\cite{brakerskiFullyHomomorphicEncryption2012,fan2012somewhat} and BGV~\cite{brakerskiLeveledFullyHomomorphic2012}, which are used in libraries such as SEAL~\cite{MicrosoftSEALRelease2023}, Lattigo~\cite{lattigo} and OpenFHE~\cite{albadawiOpenFHEOpenSourceFully2022}.
Recall that compression is compatible with modulus switching so we choose the parameters corresponding to the lowest level in a BFV/BGV parameter set. We achieve over 85\% compression and up to 98\%.





\subsubsection{Measuring Compression Key Sizes}

Using the technique described in \Cref{sec:smaller-compression-key}, we can pack the compression key and have the server unpack the key.
The unpacking can be done offline, as soon as the server receives the packed key, to reduce latency during compression.
The compression procedure is identical after the key has been unpacked, so we do not report the runtime of compression again.
Instead, we measure the size of the compression key, with and without packing, and report the time required to unpack the key.
We also distinguish two cases, non-binary and binary keys.
In the case of binary keys, we use $\delta=q+nq$ so more bits of the secret key can fit within the same ciphertext.
\Cref{tab:packed-compression-key} shows the size of the packed compression keys in the two cases. Note that even the size of the unpacked key is much smaller than commonly used cryptographic keys such as relinearization keys, automorphism keys, and bootstrapping keys, which could be as large as 100 MB.

\begin{table}[b]
    \centering
    \caption{
        Size of packed compression keys and unpacking time.
        We distinguish the case of binary and non-binary keys since binary keys can be packed more than non-binary keys.
        The Paillier modulus is 3072 bits in all cases.
    }
    \resizebox{\columnwidth}{!}{
    \begin{tabular}{c|c|c|c|c}
    \toprule
         $(n,\log_2 q)$ & (630,64) & (742,64) & (870,64) & (1305,11) \\
    \midrule
        Unpacked Key & 240 KB & 284 KB & 334 KB & 501 KB \\
    \midrule
        Packed Non-binary Key & 22 KB & 26 KB & 30 KB & 11 KB \\
        Unpacking Time & 14 ms & 25 ms & 74 ms & 15 ms \\
    \midrule
        Packed Binary Key & 12 KB & 14 KB & 16 KB & 7 KB \\
        Unpacking Time & 13 ms & 12 ms & 13 ms & 15 ms \\
    \bottomrule
    \end{tabular}
    }
    \label{tab:packed-compression-key}
\end{table}

\subsection{Batched Compression Evaluation}

\begin{figure*}[t]
    \centering
    \begin{tikzpicture}
        \begin{axis}[
            height=0.2\columnwidth,
            width=2*\columnwidth, 
            hide axis,
            xmin=0,
            xmax=1,
            ymin=0,
            ymax=1,
            legend columns=-1, 
            legend style={/tikz/every even column/.append style={column sep=0.5cm}},
            legend to name=named, 
        ]
    
            \addlegendimage{color=blue, mark size=1pt, mark=*}
            \addlegendentry{Compressed}
    
            \addlegendimage{color=black, mark size=1pt, mark=*}
            \addlegendentry{Compressed (Binary Key)}

            \addlegendimage{color=orange, mark size=1pt, mark=*}
            \addlegendentry{Rescaled \& Compressed (Binary Key)}
    
            \addlegendimage{color=teal, mark size=1pt, mark=*}
            \addlegendentry{Chen~\cite{chenEfficientHomomorphicConversion2021}}
        \end{axis}
        \end{tikzpicture}
    
\pgfplotslegendfromname{named} 

\begin{subfigure}[b]{\columnwidth}
    \centering
    \begin{tikzpicture}[
        declare function={
            logb(\x,\y) = ln(\x)/ln(\y);
            s=1; 
            log2m=3072; 
            n=630; 
            log2q=64; 
            newlog2q=20; 
        },
    ]
    \begin{axis}[
        name=plot1,
        at={(0,0)},
        xlabel = {\footnotesize Number of Compressed LWE Ciphertexts},
        ylabel = {\footnotesize Aggregate Size (KB)},
        domain=0:1000,
        samples=100,
        ymin = 0.4,
        ymax = 400,
        ymode = log,
        xmin = -50,
        xmax = 1090,
        height = 0.45\columnwidth,
        width = \columnwidth,
        legend style={at={(0.2,0.6)},anchor=east} 
    ]

        \addplot [color=red, mark size=0.75pt] {x * (n+1) * log2q / 8192 };
    
        \addplot [color=blue, mark size=0.75pt, mark=*] { (s+1)*log2m * ceil(x / floor(s * log2m / (1 + ceil(logb(n,2)) + 2*log2q))) / 8192 };

        \addplot [color=black, mark size=0.75pt, mark=*] { (s+1)*log2m * ceil(x / floor(s * log2m / (1 + ceil(logb(n,2)) + log2q))) / 8192 };

        \addplot [color=orange, mark size=0.75pt, mark=*] { (s+1)*log2m * ceil(x / floor(s * log2m / (1 + ceil(logb(n,2)) + newlog2q))) / 8192 };

        \addplot [color=teal, mark size=0.75pt, mark=*] table [
            x=slots,
            y expr=\thisrow{compressed_size_B}/1024,
            col sep=comma
        ] {data/chen.csv};
        
    \end{axis}    
    \end{tikzpicture}
    \caption{Compressed Size}
    \label{fig:batched-lwe-compression-n=630-communication}
\end{subfigure}
~
\begin{subfigure}[b]{\columnwidth}
    \begin{tikzpicture}[]
        \begin{axis}[
            name=plot1,
            at={(0,0)},
            xlabel = {\footnotesize Number of Compressed LWE Ciphertexts},
            ylabel = {\footnotesize Compression Time (s)},
            height = 0.45 \columnwidth,
            width = \columnwidth,
            ymode=log,
            legend pos=south east,
            error bars/y dir=both,
            error bars/y explicit
        ]

            \addplot [color=blue, mark size=1pt, mark=*] table [
                x=num_cts,
                y expr=\thisrow{compress_time_us_mean}/1000000,
                y error expr=\thisrow{compress_time_us_std}/1000000,
                col sep=comma
            ] {data/table2_nonbinkey.csv};

            \addplot [color=black, mark size=1pt, mark=*] table [
                x=num_cts,
                y expr=\thisrow{compress_time_us_mean}/1000000,
                y error expr=\thisrow{compress_time_us_std}/1000000,
                col sep=comma
            ] {data/table2_binkey.csv};

            \addplot [color=orange, mark size=1pt, mark=*] table [
                x=num_cts,
                y expr=\thisrow{compress_time_us_mean}/1000000,
                y error expr=\thisrow{compress_time_us_std}/1000000,
                col sep=comma
            ] {data/table2_binkey_switched.csv};

            \addplot [color=teal, mark size=1pt, mark=*] table [
                x=slots,
                y expr=\thisrow{total_time_ms}/1000,
                col sep=comma
            ] {data/chen.csv};
        \end{axis}    
    \end{tikzpicture}
    \caption{Time for Compression}
    \label{fig:batched-lwe-compression-n=630-computation}
\end{subfigure}
    \caption{
        Compressed size and compression time required for compressing LWE ciphertexts with $(n,q)=(630, 2^{64})$ using batched compression.
        The red line denotes the baseline size of uncompressed LWE ciphertexts.
    }
    \label{fig:compressing-batched-lwe}
\end{figure*}
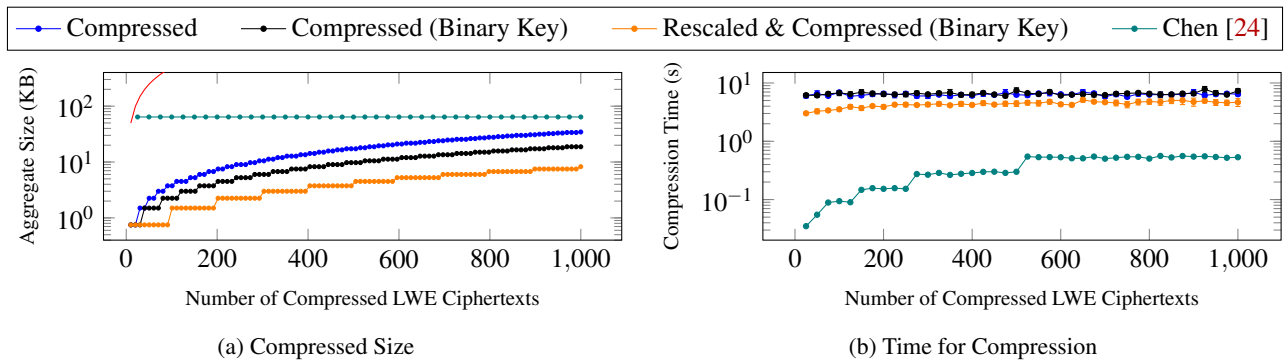

Next, we evaluate the batched compression of LWE ciphertexts.
In both cases, we use the batched compression algorithm to compress $\ell$ ciphertexts into additive Paillier ciphertexts.
Note that for batched compression, we can not use packed compression keys.

We distinguish the case of binary keys from non-binary keys in the experiment, denoted by blue and black lines, respectively.
As mentioned in \Cref{sec:batched-compression}, we set the scale to $\gamma=q + n q^2$ and $\gamma=q + n q$ in the case of non-binary and binary keys, respectively.
We also visualize a case where we rescale to a smaller modulus before compression.

The alternative approach to packing many LWE ciphertexts is RLWE packing~\cite{chenEfficientHomomorphicConversion2021,chillottiTFHEFastFully2020}.
Our reference point for RLWE packing is the work of Chen et al.~\cite{chenEfficientHomomorphicConversion2021} which is state-of-the-art in compression and has better runtime than related work.
This work maps LWE ciphertexts in $\ZZ_q^n\times\ZZ$ to an RLWE ciphertext in $\ZZ_q[X]/(X^n+1)$.

\Cref{fig:batched-lwe-compression-n=630-communication} shows the size after compression as a function of the number of compressed LWE ciphertexts.
\Cref{fig:batched-lwe-compression-n=630-computation} also shows the runtime of batched compression.
The blue and black plots correspond to non-binary and binary keys, respectively.
The orange plot also shows batched compression but over ciphertexts that are rescaled to $r=2^{20}$.
In the case of non-binary keys, we only need one Paillier ciphertext for up to 14 LWE ciphertexts and for binary keys, it is about 26.
Overall, compressing more LWE ciphertexts offers more compression compared to compressing only one LWE ciphertext.
This is because more of the plaintext space of Paillier is being utilized as more ciphertexts are compressed.
We also observe that rescaling improves the runtime and allows more compression, as is expected.
In comparison to the work of Chen et al., our compression protocol offers more than one order of magnitude better compression compared to RLWE packing.
However, our compression approach is slower than RLWE packing, particularly due to the use of expensive modular exponentiations.
In summary, in cases where there are fewer items to pack, our compression techniques offers much better compression while in cases with many LWE ciphertexts to pack, the work of Chen et al.\ is the more suitable solution.

\subsection{PIR Evaluation}
\label{sec:pir-evaluation}

In our PIR evaluation, we compare with two classes of protocols.
First, we compare with other protocols where the server stores a state for the client.
We show how the state in \protocol{} differs from that of related work and how we achieve higher throughput than similar works in this model.
We also evaluate \protocol{} as a protocol with client-side storage.

\textbf{Implementation \& Experimental Considerations.}
We implement \protocol{} in C++ using AVX512 and the libhcs library, which supports large integer arithmetic for Paillier operations.
For the LWE parameter, we use $(n,q)=(1400, 2^{32})$ and a discrete Gaussian distribution with a standard deviation of 6.4, which gives over 128-bit security as shown by the lattice-estimator~\cite{albrechtConcreteHardnessLearning2015}.
We use the same parameters for the Paillier encryption that we used in the previous section.
Offline computation, which involves Paillier operations, is parallelized but online computation is single-threaded for a fair comparison with related work.
We perform the second matrix multiplication $(\hint\cdot\ck_o)$ over an RNS basis consisting of 240 prime moduli of 27-bit length.
This allows us to perform the matrix multiplication within the machine word size, without the need to perform modular reduction until the end of the computation, significantly enhancing performance.

For our experiments, we use the reference implementation of each protocol when it is provided.
All experiments are performed 5 times and the average statistics are reported.
We conduct all experiments on an Intel(R) Xeon(R) Platinum 8276 CPU running Ubuntu 24.04.
Our code is open-source and available~\footnote{\url{https://github.com/RasoulAM/ZipPIR}}.

\subsubsection{Evaluation of PIR without Client Storage}
In this section, we compare with state-of-the-art client-efficient single-server PIR protocols.
This consists of protocols that require per-client storage on the server, i.e., SealPIR~\cite{angelPIRCompressedQueries2018a}, FastPIR~\cite{ahmadAddraMetadataprivateVoice2021}, OnionPIR~\cite{mugheesOnionPIRResponseEfficient2021}, Spiral~\cite{menonSPIRALFastHighRate2022}.
Protocols with silent preprocessing (that also support large payloads), such as YPIR (using SimplePIR)~\cite{menonYPIRHighThroughputSingleServer2024}, and HintlessPIR~\cite{liHintlessSingleServerPrivate2023} are also relevant and are included in this category.
In each protocol, we measure the per-client storage required by the server, the communication costs, and the runtime.
We report the online server throughput, a metric proposed by Henzinger et al.~\cite{henzingerOneServerPrice2023}, which shows how many database elements are processed per unit of time.
\Cref{tab:compare-pir-with-state} summarizes the results of this comparison.

\begin{table*}[t]
    \caption{
        Comparison of communication and computation costs with PIR protocols with per-client storage.
        All protocols retrieve a payload of at least 16 KB.
        * In all protocols except \protocol{}, the offline communication is equal to the per-client storage. In \protocol{}, the offline communication is 400 B.
    }
    \label{tab:compare-pir-with-state}
    \centering
\resizebox{\textwidth}{!}{
    \begin{tabular}{c|rcccccccccc} 
        \toprule
        DB & Metric & SealPIR~\cite{angelPIRCompressedQueries2018a} & FastPIR~\cite{ahmadAddraMetadataprivateVoice2021} & OnionPIR~\cite{mugheesOnionPIRResponseEfficient2021} & Spiral~\cite{menonSPIRALFastHighRate2022} & HintlessPIR~\cite{liHintlessSingleServerPrivate2023} & YPIR~\cite{menonYPIRHighThroughputSingleServer2024} & \protocol{} &  \\ 
        \midrule
        \midrule
        \multirow{6}{*}{\begin{tabular}{c}256 MB \end{tabular}}
         & Per-client Storage & 1215 KB & 670 KB & 5538 KB & 12432 KB & - & - & 60 KB \\\cmidrule{2-9} 
         & Query & 90 KB & 832 KB & 256 KB & 14 KB & 424 KB & 518 KB & 617 KB \\
         & Response & 181 KB & 64 KB & 128 KB & 20 KB & 964 KB & 120 KB & 91 KB \\
         & Total Comm. & 271 KB & 896 KB & 384 KB & 34 KB & 1388 KB & 638 KB & 708 KB \\\cmidrule{2-9} 
         & Server Time & 2455 ms & 1851 ms & 3378 ms & 1633 ms & 662 ms & 388 ms & 171 ms \\
         & Throughput & 104 MB/s & 138 MB/s & 75 MB/s & 156 MB/s & 383 MB/s & 659 MB/s & \cellcolor{MyGreen} 1497 MB/s \\
       \midrule
        \multirow{6}{*}{\begin{tabular}{c}512 MB \end{tabular}}
         & Per-client Storage & 1215 KB & 670 KB & 5538 KB & 12656 KB & - & - & 85 KB \\\cmidrule{2-9} 
         & Query & 90 KB & 1664 KB & 256 KB & 14 KB & 448 KB & 574 KB & 655 KB \\
         & Response & 181 KB & 64 KB & 128 KB & 20 KB & 964 KB & 120 KB & 128 KB \\
         & Total Comm. & 271 KB & 1728 KB & 384 KB & 34 KB & 1452 KB & 694 KB & 783 KB \\\cmidrule{2-9} 
         & Server Time & 4564 ms & 2808 ms & 5865 ms & 3031 ms & 674 ms & 425 ms & 227 ms \\
         & Throughput & 112 MB/s & 182 MB/s & 87 MB/s & 168 MB/s & 759 MB/s & 1204 MB/s & \cellcolor{MyGreen} 2255 MB/s \\
       \midrule
       \multirow{6}{*}{\begin{tabular}{c}1 GB\end{tabular}}
         & Per-client Storage & 1350 KB & 670 KB & 5538 KB & 12656 KB & - & - & 120 KB \\\cmidrule{2-9} 
         & Query & 90 KB & 3328 KB & 256 KB & 14 KB & 488 KB & 686 KB & 708 KB \\
         & Response & 181 KB & 64 KB & 128 KB & 20 KB & 1748 KB & 120 KB & 181 KB \\
         & Total Comm. & 271 KB & 3392 KB & 384 KB & 34 KB & 2236 KB & 806 KB & 889 KB \\\cmidrule{2-9} 
         & Server Time & 9071 ms & 4698 ms & 9662 ms & 3080 ms & 971 ms & 515 ms & 339 ms \\
         & Throughput & 112 MB/s & 217 MB/s & 105 MB/s & 332 MB/s & 1053 MB/s & 1988 MB/s & \cellcolor{MyGreen} 3020 MB/s \\
       \midrule
        \multirow{6}{*}{\begin{tabular}{c}2 GB\end{tabular}}
        & Per-client Storage & 1350 KB & 670 KB & 5538 KB & 14224 KB & - & - & 170 KB \\\cmidrule{2-9} 
         & Query & 90 KB & 6592 KB & 256 KB & 14 KB & 616 KB & 910 KB & 782 KB \\
         & Response & 181 KB & 64 KB & 128 KB & 20 KB & 1748 KB & 120 KB & 255 KB \\
         & Total Comm. & 271 KB & 6656 KB & 384 KB & 34 KB & 2364 KB & 1030 KB & 1037 KB \\\cmidrule{2-9} 
         & Server Time & 17692 ms & 8388 ms & 19569 ms & 4435 ms & 1232 ms & 704 ms & 580 ms \\
         & Throughput & 115 MB/s & 244 MB/s & 104 MB/s & 461 MB/s & 1661 MB/s & 2909 MB/s & \cellcolor{MyGreen} 3531 MB/s \\
        \bottomrule
    \end{tabular}
}
\end{table*}

\protocol{} has a much higher throughput than all works in this category.
This can be attributed to the simple online phase, which only consists of two matrix multiplications.
In contrast, all other protocols involve operations over RLWE ciphertexts.
Compared to the fastest alternative, YPIR, \protocol{} is at least 20\% faster in all cases, and over 2x faster for some database sizes.

The setup of \protocol{} only requires a one-time communication cost of 400 B, which consists of the clients Paillier public key and cryptographic seed.
In all instances of \protocol{}, 527 KB of the total query size consists of the compression key offsets ($\ck_o$), and the remainder is the LWE query vector ($\qu_o$).
The online communication cost of \protocol{} is similar to that of HintlessPIR and YPIR.
SealPIR, OnionPIR, and Spiral have smaller per-query communication costs, since they are highly optimized for this metric at the cost of throughput and server-side storage.

\subsubsection{Evaluation of PIR with Client-side Storage}

As described in \Cref{sec:extensions}, we can extend \protocol{} to leverage client-side storage as well.
\protocol{} with client storage requires less communication in the online phase.
Given that \protocol{} is a protocol with linear online time, we compare it with other protocols that also fall in this category.
Existing protocols with sublinear online time impose client-side costs that are inconsistent with our model.
We compare with the work of Patel et al.~\cite{patelPrivateStatefulInformation2018}, which uses a client-side state.
The authors instantiate PSIR using SealPIR, PaillierPIR, and XPIR in their paper.
They do not provide the implementation, so we compare with the numbers presented in their paper~\cite[Table 2 \& 3]{patelPrivateStatefulInformation2018}.
The result of this comparison is summarized as follows:

\begin{table}[t]
    \caption{
        Comparison of communication and computation costs with PSIR variants~\cite{PrivateStatefulInformation} over a database with $n$ items of size 288~B.
        In the case of \protocol{}, the stated client state size supports one query for that client.
    }
    \label{tab:comparison-statefulPIR}
    \centering
\resizebox{\columnwidth}{!}{
    \begin{tabular}{c|rcccc}
        \toprule
        \textbf{$n$} & \textbf{Metric} & {\small\textbf{SealPSIR}} & {\small\textbf{PaillierPSIR}} & {\small\textbf{XPSIR}} & \textbf{ZipPIR} \\
        \midrule
        \midrule
        \multirow{5}{*}{$2^{16}$}
            & Client Storage & 129 KB     & 129 KB     & 129 KB     & 6 KB \\\cmidrule{2-6}
            & Query & 74 KB      & 12 KB      & 305 KB     & 591 KB \\
            & Response & 256 KB     & 18 KB      & 262 KB     & 3 KB \\
            & Online Total & 330 KB     & 30 KB      & 567 KB     & 594 KB \\\cmidrule{2-6}
            & Online Server & 70 ms      & 760 ms     &  20 ms      & \cellcolor{MyGreen} 16 ms \\
        \midrule
        \multirow{5}{*}{$2^{18}$}
            & Client Storage & 258 KB     & 258 KB     & 258 KB     & 12 KB \\\cmidrule{2-6}
            & Query & 84 KB      & 24 KB      & 413 KB     & 655 KB \\
            & Response & 256 KB     & 18 KB      & 262 KB     & 6 KB \\
            & Online Total & 340 KB     & 42 KB      & 675 KB     & 661 KB \\\cmidrule{2-6}
            & Online Server & 210 ms     & 1020 ms    & 120 ms     & \cellcolor{MyGreen} 40 ms \\
        \midrule
        \multirow{5}{*}{$2^{20}$}
            & Client Storage & 524 KB     & 524 KB     & 524 KB     & 24 KB \\\cmidrule{2-6}
            & Query & 104 KB     & 48 KB      & 597 KB     & 783 KB \\
            & Response & 256 KB     & 18 KB      & 262 KB     & 12 KB \\
            & Online Total & 360 KB     & 66 KB      & 895 KB     & 795 KB \\\cmidrule{2-6}
            & Online Server & 710 ms     & 1750 ms    & 520 ms     & \cellcolor{MyGreen} 120 ms \\
        \bottomrule
    \end{tabular}
}
\end{table}

\textbf{Online Costs.}
In almost all instances, \protocol{} is faster than all instantiations of PSIR.
This can be attributed to the fact that PSIR performs PIR as a subroutine, which necessitates the use of public key operations, even though the number of such operations is sub-linear in the database size.
In contrast, \protocol{} only performs plaintext operations in the online phase.

With regards to network costs, \protocol{} requires less online communication compared to XPSIR, which is the fastest variant of PSIR.
The communication cost is comparable to that of SealPSIR and strictly worse than PaillierPSIR.
However, these two protocols are much slower, as mentioned before, which shows a communication-computation tradeoff.

\textbf{Offline Costs.}
In PSIR, the database is streamed to the client, and the state, which is sublinear in the database size, is stored.
In \protocol{} with client-side storage, only the hint is sent to the client, which is much smaller than the database and the PSIR hint, but only supports one query.



%% file: discussion.tex
\section{Discussion}


\textbf{Why not only use Paillier?}
We construct \protocol{} as a combination of LWE-based PIR and Paillier.
One might suggest using only Paillier for the full construction.
Corrigan-Gibbs and Kogan~\cite[Appendix E.2]{10.1007/978-3-030-45721-1_3} also suggested a solution that uses additive homomorphic encryption.
However, such a protocol requires a large number of Paillier plaintext multiplications, linear in the database size.
While this expensive step is performed in the offline phase, it is required to be done once for each query, which is a high computational cost, even in the offline phase.
In contrast, the offline phase of \protocol{} only requires a sublinear number of Paillier operations, which is much more practical.

\textbf{ZipPIR+DoublePIR.}
We can also construct a protocol, similar to \protocol{}, that builds on top of DoublePIR instead of SimplePIR.
In this variant, the size of the hint would be reduced and the matrix multiplication between the hint and the compression key offset (\Cref{zippir:online-compression-offset}) would be faster.
This would, however, limit the size of the payload to only one byte, similar to DoublePIR and other protocols built on it, such as YPIR~\cite{menonYPIRHighThroughputSingleServer2024}.
We leave a comparison of protocols focused on small payloads~\cite{10.1145/3658644.3690328, menonYPIRHighThroughputSingleServer2024, henzingerOneServerPrice2023} for future work.

\textbf{\protocol{} using a Packed Compression Key.}
\protocol{} can also be modified to use the packed compression key from \Cref{sec:lwe-compress-packed-keys}.
Given that the bulk of the query cost consists only of the compression key, this can reduce the query size at the cost of a larger response. 

\textbf{Advantages for database updates.}
\protocol{} delegates all offline storage and computation to the server (except constant storage for keys and the PRNG seed).
The benefit of this model is that database updates require no interaction with the client.
This is compatible with our assumption that the client is resource-constrained.
It is also efficient from a bandwidth standpoint since the server does not need to communicate with all available clients upon each database update, which is not scalable if many clients query the server.
The disadvantage is that client-specific hints must be updated when the database changes, but this can be mitigated with additional server resources, without any client involvement.

%% file: applications.tex




\section{Private SCT Auditing at Fixed Intervals}
Henzinger et al. proposed the use of PIR in context of certificate transparency, specifically for auditing of signed certificate timestamps (SCT)~\cite{henzingerOneServerPrice2023}.
The authors build on Google's recent solution of using an SCT auditor~\cite{optOutSCTDeBlasio}.

The initial solution of Henzinger et al. provided high throughput but came at the cost of a large hint stored by the clients.
The problem is that this hint must be updated upon each change in the database.
To combat this problem, the authors suggested that clients download hints periodically.

We propose a different approach to handle database updates, while maintaining throughput similar to that of SimplePIR.
Typically, the set of clients performing SCT checks are known to the auditor (e.g. Google), so we assume that clients can register to perform tests and submit their public keys.
Moreover, clients collect SCTs from TLS handshakes during browsing and perform an audit in a batch.
So in our proposed solution, we assume that clients audit a fixed-sized batch of SCTs at regular intervals.
Using this assumption, the server can perform precomputation in anticipation of the client's request.
The precomputation can be done on the latest version of the database, hence allowing to detect malicious behavior sooner.


%% file: related-work.tex
\section{Related Work}
\label{sec:related-work}

\subsection{Related work on compression}

Compression of homomorphic ciphertexts is achieved by performing scheme-switching, i.e., changing the scheme under which the message is encrypted.

\textbf{Compression from Precise Scheme-switching.}
Techniques such as modulus switching~\cite{cheonHomomorphicEncryptionArithmetic2017a} and dimension reduction~\cite{brakerskiEfficientFullyHomomorphic2011} change the parameters of LWE and RLWE ciphertexts, resulting in smaller ciphertexts.
LWE-to-RLWE packing~\cite{chenEfficientHomomorphicConversion2021} compresses a list of LWE ciphertexts into one RLWE ciphertext, which reduces the total size of the ciphertexts.
However, this method is less effective if less than $N$ ciphertexts are compressed, where $N$ is the degree of the RLWE ciphertext.
Coefficient extraction performs the inverse functionality of LWE-to-RLWE packing, i.e., extracting specific coefficients of the RLWE ciphertext and discarding the rest, which can reduce the overall size.

\textbf{Compression from Imprecise Scheme-switching.}
Imprecise scheme-switching changes the scheme of the ciphertext, but encrypts to a message that is slightly different than the original message.
Imprecise scheme-switching is sufficient if the final result is decrypted and no further computation is performed.

Hu~\cite{huImprovingEfficiencyHomomorphic2013} introduced the concept of \textit{secure converters} for converting between cryptographic schemes.
This is achieved by homomorphically evaluating (part of) the decryption circuit of the source scheme under the destination scheme. 
Within that framework, the author proposed homomorphically converting from LTV ciphertexts\cite{10.1145/2213977.2214086} to Paillier ciphertexts to reduce bandwidth usage from the server to the client.
The conversion, however, is not precise and the Paillier ciphertexts encrypt a noisy version of the initial plaintexts.
Using this approach, a 256x compression rate is achieved whilst communicating ciphertexts back to the client.
However, the LTV cryptographic scheme is not adopted as a practical homomorphic encryption scheme.

Brakerski et al.~\cite{brakerskiLeveragingLinearDecryption2019} showed how a high-rate compression, arbitrarily close to one, can be achieved over ciphertexts with the \textit{linear-decrypt-and-multiply} characteristic.
Cryptosystems with linear-decrypt-and-multiply can decrypt to any multiple of the message.
Based on the authors, among prevalent encryption schemes, only GSW falls into that category.
Assuming the goal is to encrypt $\{m_0,m_1,\cdots,m_{\ell-1}\}$, then the compression is done by homomorphically decrypting these messages to $\{m_0+e_0,\Delta m_1+e_1,\cdots,\Delta^{\ell-1} m_{\ell-1}+e_{\ell-1}\}$, where $e_i$'s are noise introduced from the homomorphic cryptosystem, similar to LWE.
By adding these messages together, the server obtains one large plaintext, encrypted under an additive ciphertext which it sends to the client.

Gentry et al.~\cite{gentryCompressibleFHEApplications2019} also proposed a method to compress many GSW ciphertexts into high-rate PVW~\cite{10.1007/978-3-540-85174-5_31} ciphertexts.
The ratio between the plaintext and ciphertext can be arbitrarily close to one in their construction.
However, this can only be achieved if the underlying aggregate plaintext is very large.
Specifically, for the ratio to be $1-\epsilon$, the aggregate plaintext must be proportional to $1/\epsilon^3$.
The authors described how to construct a PIR protocol from this technique, but their compression technique is not applicable to any other type of ciphertext.

%% file: related-work-pir.tex
\subsection{Related Work on PIR}
Computational PIR (CPIR) protocols follow one of two approaches:
1) the server (with or without the help of the client) computes a \textit{hint} that is given to the client 2) the client sends cryptographic keys to the server which are stored.
We describe each approach briefly, the advantages and disadvantages of each approach and list related work.

\subsubsection{PIR with client-storage}
One approach to PIR is to compute a database-dependent hint which is stored by the client before issuing queries.
This hint can speed up subsequent queries.
SimplePIR~\cite{henzingerOneServerPrice2023} and FrodoPIR~\cite{davidsonFrodoPIRSimpleScalable2023} propose a PIR protocol based on LWE with a client-independent hint.
The hint size is $O(\sqrt{N}n)$ for $N$ database rows and LWE dimension of $n$.
All clients use the same hint, which helps quickly respond to PIR queries and achieves high throughput (up to 10 GB/s).
However, the hint is a high upfront network cost (100 MB for a 1 GB database), requires large storage by the client, and must be recalculated and redistributed to the clients every time the database is updated.
DoublePIR extends SimplePIR such that the hint is smaller, but the overall throughput is less and it is limited to retrieval of one byte of information per query.

Another similar line of work proposes PIR with online time that is sublinear in the size of the database~\cite{corrigan-gibbsSingleServerPrivateInformation2022,piano}.
The client stores some information in the form of a hint, which is used to issue queries and must be updated after a certain number of queries have been made.
Despite the very high throughput of these protocols, the requirements are highly impractical for a resource-constrained client.
In some protocols~\cite{piano}, the server streams the entire database to the client in the offline phase, and the storage requirements for the client are sublinear in the size of the database.

\subsubsection{PIR with per-client server-side storage}
Another category of works assumes auxiliary information is sent before the start of the protocol, usually in the form of cryptographic keys.
The cost of sending these keys is amortized over many queries, but requires per-client storage on the server.
The per-query communication costs in such protocols are small, and if sufficient queries are made, the runtime and communication cost of setup are amortized.
Works that follow this model include SealPIR~\cite{angelPIRCompressedQueries2018a}, MulPIR~\cite{Ali2019CommunicationComputationTI}, OnionPIR~\cite{mugheesOnionPIRResponseEfficient2021}, Constant-weight PIR~\cite{mahdaviConstantweightPIRSingleround2022}, Pantheon~\cite{ahmadPantheonPrivateRetrieval2022}, FastPIR~\cite{ahmadAddraMetadataprivateVoice2021}, Spiral (and its variants)~\cite{menonSPIRALFastHighRate2022}, and SparsePIR~\cite{patelDonBeDense2023}.

HintlessPIR~\cite{liHintlessSingleServerPrivate2023} and YPIR~\cite{menonYPIRHighThroughputSingleServer2024} remove the need for client-specific storage and generate one large preprocessed hint, which depends on the database. This hint, which is stored by the server, is used to respond to all client queries.

%% file: appendix.tex
\appendix

\section{LWECompress Using Packed Compression Keys}
\label{sec:lwe-compress-packed-keys}

The necessary procedures for using a packed compression key is given in \Cref{alg:packed-key-compress}. The $\textsc{GeneratePackedKey}$ procedure generates the packed key from the LWE secret key.
The unpacking procedure computes the compression key from the packed compression key by scaling the packed key with different values.
By packing in this particular manner, the compression procedure can be done similar to before, without any changes.
The final change is made in the decryption.
The response is not necessarily in the lower order bits of the additive ciphertext ciphertext anymore, so a division is required before continuing with the rest of the LWE decryption procedure.

\begin{algorithm}[t]
	 \caption{Procedures for using a packed compression key,       including generating the packed key, unpacking it, and the corresponding modified decryption function.
        }
	 \label{alg:packed-key-compress}
	 \begin{algorithmic}[1]
    \Procedure{GeneratePackedKey}{$\addkey,\sk$}
        \State $t = \floor{\frac{0.5\log_2 m}{\log_2\delta}}$
        \For {$i\in[\ceil{n/t}]$}
            \State $r \leftarrow \delta^{-(t-1)} (\sum_{j\in[t]}\sk[it+j] \cdot \delta^{j}) \mod m$
            \State $ \pck_{i} \leftarrow \texttt{AEnc}(\addkey, r)$  
        \EndFor
        \State \Return $\pck$
    \EndProcedure
    \vspace{3mm}
    \Procedure{UnpackCompressionKey$_q$}{$\pck$}
        \For {$i\in[\ceil{n/t}]$}
            \For {$j\in[t]$}
                \State $\ck[it+j] \leftarrow \delta^{t-1-j} \otimes \pck[i]$
            \EndFor
        \EndFor
    \State \Return $\ck$
    \EndProcedure
    \vspace{3mm}
    \Procedure{ModifiedLWEDecryptPackedKey$_{q, p}$}{$\addkey,x$}
        \State $y \leftarrow \delta^{(t-1)}  \texttt{ADec}(\addkey, x) \mod m$ 
        \State $\mu^{**} = \floor{y / \delta^{(t-1)}} \mod q$
        \vspace{1mm}
        \State $ \mu'' = \lfloor \mu^{**}/\Delta\rceil$
        \Comment{$\Delta=\round{q/p}$}
        \vspace{2mm}
        \State \Return $\mu'' \in \ZZ_{p}$
   \EndProcedure
    \end{algorithmic}
\end{algorithm}

\section{Compressig RLWE Coefficients}
\label{sec:rlwe-compressing}

Similar to LWE, we can also construct an encryption scheme based on the Ring Learning with Errors (RLWE)~\cite{lyubashevskyIdealLatticesLearning2010} assumption, which we will denote as $\mathcal{E}_{\text{RLWE}}$. 

\begin{algorithm}[t]
     \caption{Encryption and Decryption of $\mathcal{E}_{RLWE}$}
     \label{alg:rlwe-encrypt-decrypt}
     \begin{algorithmic}[1]
        \Procedure{RLWEEncrypt}{$S(X), \mu(X)$}
        
        \State Sample $A(X)\xleftarrow{\$} R_q$ and $E(X) \leftarrow \chi$
        \vspace{2mm}
        \State $B(X) = A(X)\cdot S(X) + \Delta \cdot \mu(X) + E(X) \mod R_q$
        \State \Return $C=(A(X), B(X))$
        \EndProcedure

        \vspace{3mm}
        
        \Procedure{RLWEDecrypt}{$S(X), C$}
        \State $(A(X), B(X)) \leftarrow C$
        \vspace{1mm}
        \State $\mu^*(X) = (B(X) - A(X)\cdot S(X)) \mod R_q$
        \vspace{2mm}
        \State $\mu'(X)=\lfloor\mu^*(X)/\Delta\rceil$
        \vspace{2mm}
        
        \State \Return $\mu'(X)\in R_p$
        \EndProcedure     
     \end{algorithmic}
\end{algorithm}
\label{sec:rlwe-compression}
Cryptosystems such as BGV~\cite{brakerskiLeveledFullyHomomorphic2012}, BFV~\cite{brakerskiFullyHomomorphicEncryption2012,fan2012somewhat}, and CKKS~\cite{cheonHomomorphicEncryptionArithmetic2017a} have ciphertexts of a similar format.
RLWE ciphertexts are useful since they can encrypt a polynomial, i.e. a vector of numbers, instead of just one scalar. For RLWE encryption, we select a dimension $N$, ciphertext modulus $q$, plaintext modulus $p$, and $\Delta=\round{q/p}$. Define $R_q=\ZZ_q[X]/(X^N+1)$ and $R_p$ similarly. Moreover, define a discrete error distribution $\chi$ over $R_q$.
For key generation, sample $S(X)$ uniformly from $R_q$.

Similar to LWE, we can also compress fresh RLWE ciphertexts by sending the seed used to generate $A(X)$~\cite{aliCommunicationComputationTradeoffs2021a, MicrosoftSEALRelease2023, albadawiOpenFHEOpenSourceFully2022}. Using this technique, the size of a ciphertext can be reduced from $2N\log_2 q$ bits to $\lambda + N\log_2 q$.

RLWE ciphertexts also have a linear phase evaluation and hence, can benefit from our compression technique.
However, an RLWE ciphertext is only twice as large as the phase so the compression technique, applied naively, would not yield a significant improvement.
Our approach is beneficial if the user is only interested in some coefficients of the RLWE plaintext and not all of them.

The main observation is that each coefficient of $\mu'(X)$ in the \textsc{RLWEDecrypt} procedure can be calculated separately. Specifically, for $0\leq k \leq N-1$

\makeatletter
    \def\tagform@#1{\maketag@@@{\normalsize(#1)\@@italiccorr}}
\makeatother

{\tiny
\begin{align}
    \mu'&[k] = \lfloor\frac{\mu^*[k]}{\Delta}\rceil \\
    &= \left\lfloor \frac{B[k] - \sum_{i=0}^{k} A[k-i] \cdot S[i] + \sum_{i=k+1}^{N-1} A[N+k-i] \cdot S[i]}{\Delta} \right\rceil
    \label{eq:rlwe-extract-general}
\end{align}
}%

Note that the operations in the numerator are happening modulo $q$. The numerator of \Cref{eq:rlwe-extract-general} is a linear combination of the coefficients of the secret key, hence it can be computed given the encrypted coefficients of the secret key. The complete procedure to compress the coefficient of $X^k$ in an RLWE plaintext and the corresponding decryption function is shown in \Cref{alg:rlwe-compress-response}. Compression of RLWE coefficients is fully compatible with the compact compression keys of \Cref{sec:smaller-compression-key} and batched compression of \Cref{sec:batched-compression}.

\begin{algorithm}[H]
    \caption{
      Compressing Extracted RLWE Coefficient, performed by the server and the corresponding modified decryption process, for the client.
      The compression key is $\ck$ such that $\ck[i]={\texttt{AEnc}_{\texttt{s}}(S[i])}$ and $C\in R_q\times R_q$
    }
	 \label{alg:rlwe-compress-response}
	 \begin{algorithmic}[1]
    \Procedure{RLWECompressCoefficient}{$\ck, C, k$}
        \State $x=B[k]$
        \For{$i \in \{0,1,\cdots,k\}$}
            \State $x \leftarrow x \oplus \left( (q-A[k-i]) \otimes \ck[i]\right)$
        \EndFor
        \For{$i \in \{k+1,\cdots,N-1\}$}
            \State $x \leftarrow x \oplus \left(A[N+k-i] \otimes \ck[i]\right)$
        \EndFor
        
        \State \Return $x$ 
    \EndProcedure
    \vspace{3mm}
    \Procedure{ModifiedRLWEDecrypt$_{q,p}$}{$\addkey, x$}
        \State $\mu^{**}_k = \texttt{ADec}(\addkey, x) \mod q$
        \vspace{1mm}
        \State $\mu''_k= \lfloor \mu^{**}_k / \Delta \rceil$
        \Comment{$\Delta=\round{q/p}$}
        \vspace{2mm}
        
        \State \Return  $\mu''_k \in \ZZ_{p}$
    \EndProcedure
	 \end{algorithmic}
\end{algorithm}

\begin{theorem}[Correctness]
\label{thm:rlwe-compress-correct}
    If $m > q + N q^2$, \Cref{alg:rlwe-compress-response} produces a compressed ciphertext which decrypts to the coefficient of $X^k$ if decrypted using \textsc{ModifiedRLWEDecrypt}$_{q,p}$. More formally, 
    \begin{align*}\normalfont
        x \leftarrow\textsc{RLWECompressCoefficient}(\ck, c, k) \\
        \mu_k'' \leftarrow\textsc{ModifiedRLWEDecrypt}_{q,p}(\texttt{s}, x)
    \end{align*}
    then $\mu_k''$ is equal to the coefficient of $X^k$ in 
    \begin{align*}
        \mu'(X) = \textsc{RLWEDecrypt}(\sk, c)
    \end{align*}
\end{theorem}

We provide the proof of \Cref{thm:rlwe-compress-correct} in \appsection{sec:prove-rlwe-compress}.
Similar to the case of LWE ciphertexts, if the coefficients of the RLWE secret key are binary, we can tighten the condition on $m$ in \Cref{thm:rlwe-compress-correct} such that $m > q + N q$.
The following corollary summarizes this fact.

\begin{corollary}
    If the coefficients of the secret key are binary and $m > q + N q$, \Cref{alg:rlwe-compress-response} produces a compressed ciphertext which decrypts to the coefficient of $X^k$ if decrypted using \textsc{ModifiedRLWEDecrypt}$_{q,p}$.
\end{corollary}

\textit{Security.} A similar argument can be made about the security of compression over RLWE. Let $\mathcal{E}''$ denote the cryptosystem which is the combination of $\mathcal{E}_{RLWE}$ and $\mathcal{E}_{A}$. The following proposition holds regarding security.

\begin{proposition}[Security]
    If $\mathcal{E}_{RLWE}$ and $\mathcal{E}_A$ are semantically secure, then $\mathcal{E}''$ is also semantically secure.
\end{proposition}

\section{Proof of RLWE Compression}
\label{sec:prove-rlwe-compress}

\begin{proof}

Line 1 of \Cref{alg:rlwe-compress-response} computes 
$$
B[k] + \sum_{i=0}^{k} (q-A[k-i]) \cdot S[i] + \sum_{i=k+1}^{N-1} A[N+k-i] \cdot S[i]\
$$
encrypted under additive encryption, which is possible due to the linear properties.
We know that all coefficients of $A(X)$, $B(X)$, and $S(X)$ are elements in $\ZZ_q$, hence

    \begin{align*}
        B[k] + \left(\sum_{i=0}^{k} (q-A[k-i]) \cdot S[i]\right)
        + \left(\sum_{i=k+1}^{N-1} A[N+k-i] \cdot S[i]\right)\\
        \leq q + \left(\sum_{i=0}^{k} q \cdot q\right) + \left(\sum_{i=k+1}^{N-1} q \cdot q\right)
        = q + Nq^2 < m
    \end{align*}
so there is no overflow in the plaintext space of the additive cryptosystem.

{\footnotesize
    \begin{align*}
        & \mu^{**}_{k} = \texttt{ADec}_{s}(x) \mod q \\
        & = \left(\left(B[k] + \sum_{i=0}^{k} (q-A[k-i]) \cdot S[i] \right. \right.\\
        & ~~~~~~~~~~~~~~~~~~~~~~~~~~ \left.\left.+\sum_{i=k+1}^{N-1} A[N+k-i] \cdot S[i]\right) \mod m \right) \mod q \\
        & = \left( B[k] + \sum_{i=0}^{k} (q-A[k-i]) \cdot S[i] + \sum_{i=k+1}^{N-1} A[N+k-i] \cdot S[i]\right) \mod q \\
        & = B[k] - \sum_{i=0}^{k} A[k-i] \cdot S[i] + \sum_{i=k+1}^{N-1} A[N+k-i] \cdot S[i] \mod q
    \end{align*}
}%

which is equivalent to the $k^{th}$ coefficient of 
$$
    \mu^*(X) = B(X) - A(X) \cdot S(X) \mod R_q
$$
which can be seen by expanding the equation.
Given that line 16 of \Cref{alg:lwe-encrypt-decrypt} performs rounding coefficient-wise, it produces the same result as line 10 of \Cref{alg:rlwe-compress-response}.
\end{proof}

\section{Security \& Correctness of \protocol{}}
\label{appendix:proof}

In this section, we prove the correctness and security of \protocol{}.

\protocolCorrectness*

\begin{proof}

We know that $\ck_r$ denotes a valid ciphertext, if and only if $\ck_r\in(\ZZ^{*}_{m^2})^{d_0}$. 
This happens with probability, $(\phi(m^2)/m^2)^{d_0} > 1 - 2d_0/\sqrt{m}$, which is high for secure chioces of $m$.
Assuming that $\ck_r\in(\ZZ^{*}_{m^2})^{d_0}$, then assume $r = \textsc{PaillierDecrypt}(\paillierkey, \ck_r)$ so
\begin{align}
    \mu &= t + \textsc{PaillierDecrypt}(\paillierkey, \clienthint) \mod m\\
    &= b + \hint \cdot \ck_{o} + \hint \cdot \textsc{PaillierDecrypt}(\paillierkey, \ck_r) \mod m \\
    &= b + \hint \cdot (\ck_{o} + r) = b + \hint \cdot \sk \mod m = b + \hint \cdot \sk
\end{align}
Where the last equation comes from the fact that $ || b+\hint \cdot \sk ||_{\infty} < q + nq < m$, so it does overflow mod $m$. Hence, 
\begin{align}
    \mu \mod q
    & = b + \hint \cdot \sk \mod q \\
    & = \db^{T} \cdot \qu_0 -  \db^{T} \cdot A \cdot \sk \mod q\\
    & = \db^{T} (\Delta u_0 + e) = \Delta \db[i] + \db^{T} e 
\end{align}

Using the same analysis from the proof of SimplePIR~\cite[Appendix C.2] {henzingerOneServerPrice2023}, we can show that $||\db^{T} e||_\infty < \Delta / 2$ with probability $1-\delta$, assuming \Cref{eq:correctness-condition} holds.
Combining this step with the previous step, we can see that with probability $(1-\delta)(1 - 2d_0/\sqrt{m}) > 1 - \delta - 2d_0/\sqrt{m}$, we will have $f = \round{(\mu \mod q) p / q} = \db[i]$, which proves the theorem.
\end{proof}

\protocolSecurity*

\begin{proof}
    We denote $(q_h, \qu)=(\pk_P, \ck_r, \ck_o, \qu_0)$ for simplicity.
    To prove the theorem, we prove that the tuple, $(\pk_P, \ck_r, \ck_o, \qu_0)$, as generate by \protocol{} is indistinguishable from a random tuple.
    We will prove this via multiple hybrid arguments. For this, we define 5 distributions to generate the tuple and show that each consecutive pair of distributions are indistinguishable. We define each distribution by the modification to the original definition of the tuple. 

    \textbf{Dist.} $\mathbf{\mathcal{H}_0}$\textbf{:} The tuple is generated as described in \protocol{}
    
    \textbf{Dist.} $\mathbf{\mathcal{H}_1}$\textbf{:} Sample $\pt_r\sample\ZZ_m$ and set $\ck_r = \enc(\paillierkey, \pt_r)$
    
    \textbf{Dist.} $\mathbf{\mathcal{H}_2}$\textbf{:} Set $\ck_r = \enc(\paillierkey, \sk)$ and sample $\ck_o$ randomly
    
    \textbf{Dist.} $\mathbf{\mathcal{H}_3}$\textbf{:} Sample $\ck_r$ and $\ck_o$ randomly
    
    \textbf{Dist.} $\mathbf{\mathcal{H}_4}$\textbf{:} Sample $\qu_0$, $\ck_r$, and $\ck_o$

    We show that each consecutive pair of distributions are indistinguishable up to a negligible factor.

    The only difference between $\mathbf{\mathcal{H}_0}$ and $\mathbf{\mathcal{H}_1}$ is the probability of producing invalid ciphertexts in $\mathbf{\mathcal{H}_0}$, which is $2d_0/\sqrt{m}$ and is negligible in the security parameter for secure choice of $m$.

    $\mathbf{\mathcal{H}_1}$ and $\mathbf{\mathcal{H}_2}$ can be simulated by a simulator which $S$ which given input from $\mathbf{\mathcal{H}_2}$ computes $S(\pt_r, \ck_r, \ck_o) = (\pt_r, \pt_r - \ck_r, \ck_o + \pt_r)$. The output of $S$ is identical to the output of $\mathbf{\mathcal{H}_1}$.

    Indistinguishability between $\mathbf{\mathcal{H}_2}$ and $\mathbf{\mathcal{H}_3}$ reduces to the IND-CPA security of the Paillier encryption.

    Lastly, indistinguishability between $\mathbf{\mathcal{H}_3}$ and $\mathbf{\mathcal{H}_4}$ is identical to that of SimplePIR~\cite[Lemma C.2]{henzingerOneServerPrice2023}.

    Combining these steps proves that the correctly generated tuple is indistinguishable from random, up to a negligible parameter.

\end{proof}

%% file: references.bib
@inproceedings{10.1007/978-3-030-45721-1_3,
  title = {Private {{Information Retrieval}} with {{Sublinear Online Time}}},
  booktitle = {Advances in {{Cryptology}} -- {{EUROCRYPT}} 2020},
  author = {{Corrigan-Gibbs}, Henry and Kogan, Dmitry},
  editor = {Canteaut, Anne and Ishai, Yuval},
  year = {2020},
  pages = {44--75},
  publisher = {Springer International Publishing},
  address = {Cham},
  isbn = {978-3-030-45721-1}
}

@inproceedings{199325,
  title = {Unobservable {{Communication}} over {{Fully Untrusted Infrastructure}}},
  booktitle = {12th \{\vphantom\}{{USENIX}}\vphantom\{\} {{Symposium}} on {{Operating Systems Design}} and {{Implementation}} (\{\vphantom\}{{OSDI}}\vphantom\{\} 16)},
  author = {Angel, Sebastian and Setty, Srinath},
  year = {2016},
  month = nov,
  pages = {551--569},
  publisher = {\{USENIX\} Association},
  address = {Savannah, GA},
  isbn = {978-1-931971-33-1}
}

@article{aguilar-melchorXPIRPrivateInformation2016,
  title = {{{XPIR}} : {{Private Information Retrieval}} for {{Everyone}}},
  shorttitle = {{{XPIR}}},
  author = {{Aguilar-Melchor}, Carlos and Barrier, Joris and Fousse, Laurent and Killijian, Marc-Olivier},
  year = {2016},
  journal = {Proceedings on Privacy Enhancing Technologies},
  issn = {2299-0984},
  urldate = {2024-02-01}
}

@inproceedings{ahmadAddraMetadataprivateVoice2021,
  title = {Addra: {{Metadata-private}} Voice Communication over Fully Untrusted Infrastructure},
  shorttitle = {Addra},
  booktitle = {15th \{\vphantom\}{{USENIX}}\vphantom\{\} {{Symposium}} on {{Operating Systems Design}} and {{Implementation}} (\{\vphantom\}{{OSDI}}\vphantom\{\} 21)},
  author = {Ahmad, Ishtiyaque and Yang, Yuntian and Agrawal, Divyakant and Abbadi, Amr El and Gupta, Trinabh},
  year = {2021},
  urldate = {2023-09-12},
  isbn = {978-1-939133-22-9},
  langid = {english}
}

@article{ahmadPantheonPrivateRetrieval2022,
  title = {Pantheon: {{Private Retrieval}} from {{Public Key-Value Store}}},
  shorttitle = {Pantheon},
  author = {Ahmad, Ishtiyaque and Agrawal, Divyakant and Abbadi, Amr El and Gupta, Trinabh},
  year = {2022},
  month = dec,
  journal = {Proceedings of the VLDB Endowment},
  volume = {16},
  number = {4},
  pages = {643--656},
  issn = {2150-8097},
  doi = {10.14778/3574245.3574251},
  urldate = {2023-08-29}
}

@inproceedings{akhavanmahdaviLevelPrivateNonInteractive2023,
  title = {Level {{Up}}: {{Private Non-Interactive Decision Tree Evaluation}} Using {{Levelled Homomorphic Encryption}}},
  shorttitle = {Level {{Up}}},
  booktitle = {Proceedings of the 2023 {{ACM SIGSAC Conference}} on {{Computer}} and {{Communications Security}}},
  author = {Akhavan Mahdavi, Rasoul and Ni, Haoyan and Linkov, Dimitry and Kerschbaum, Florian},
  year = {2023},
  month = nov,
  series = {{{CCS}} '23},
  pages = {2945--2958},
  publisher = {Association for Computing Machinery},
  address = {New York, NY, USA},
  doi = {10.1145/3576915.3623095},
  urldate = {2024-01-27},
  isbn = {9798400700507}
}

@inproceedings{albadawiOpenFHEOpenSourceFully2022,
  title = {{{OpenFHE}}: {{Open-Source Fully Homomorphic Encryption Library}}},
  shorttitle = {{{OpenFHE}}},
  booktitle = {Proceedings of the 10th {{Workshop}} on {{Encrypted Computing}} \& {{Applied Homomorphic Cryptography}}},
  author = {Al Badawi, Ahmad and Bates, Jack and Bergamaschi, Flavio and Cousins, David Bruce and Erabelli, Saroja and Genise, Nicholas and Halevi, Shai and Hunt, Hamish and Kim, Andrey and Lee, Yongwoo and Liu, Zeyu and Micciancio, Daniele and Quah, Ian and Polyakov, Yuriy and R.V., Saraswathy and Rohloff, Kurt and Saylor, Jonathan and Suponitsky, Dmitriy and Triplett, Matthew and Vaikuntanathan, Vinod and Zucca, Vincent},
  year = {2022},
  month = nov,
  series = {{{WAHC}}'22},
  pages = {53--63},
  publisher = {Association for Computing Machinery},
  address = {New York, NY, USA},
  doi = {10.1145/3560827.3563379},
  urldate = {2023-09-29},
  isbn = {978-1-4503-9877-0}
}

@incollection{albrechtCiphersMPCFHE2015,
  title = {Ciphers for {{MPC}} and {{FHE}}},
  booktitle = {Advances in {{Cryptology}} -- {{EUROCRYPT}} 2015},
  author = {Albrecht, Martin R. and Rechberger, Christian and Schneider, Thomas and Tiessen, Tyge and Zohner, Michael},
  editor = {Oswald, Elisabeth and Fischlin, Marc},
  year = {2015},
  volume = {9056},
  pages = {430--454},
  publisher = {Springer Berlin Heidelberg},
  address = {Berlin, Heidelberg},
  doi = {10.1007/978-3-662-46800-5_17},
  urldate = {2022-01-15},
  isbn = {978-3-662-46799-2 978-3-662-46800-5},
  langid = {english}
}

@misc{albrechtConcreteHardnessLearning2015,
  title = {On the Concrete Hardness of {{Learning}} with {{Errors}}},
  author = {Albrecht, Martin R. and Player, Rachel and Scott, Sam},
  year = {2015},
  number = {2015/046},
  urldate = {2024-01-02}
}

@article{Ali2019CommunicationComputationTI,
  title = {Communication-{{Computation Trade-offs}} in {{PIR}}},
  author = {Ali, Asra and Lepoint, Tancr{\`e}de and Patel, S and Raykova, Mariana and Schoppmann, Phillipp and Seth, Karn and Yeo, Kevin},
  year = {2019},
  journal = {IACR Cryptol. ePrint Arch.},
  volume = {2019},
  pages = {1483}
}

@inproceedings{aliCommunicationComputationTradeoffs2021a,
  title = {Communication--{{Computation Trade-offs}} in {{PIR}}},
  booktitle = {30th {{USENIX Security Symposium}} ({{USENIX Security}} 21)},
  author = {Ali, Asra and Lepoint, Tancr{\`e}de and Patel, Sarvar and Raykova, Mariana and Schoppmann, Phillipp and Seth, Karn and Yeo, Kevin},
  year = {2021},
  pages = {1811--1828},
  urldate = {2023-09-28},
  isbn = {978-1-939133-24-3},
  langid = {english}
}

@inproceedings{angelPIRCompressedQueries2018a,
  title = {{{PIR}} with {{Compressed Queries}} and {{Amortized Query Processing}}},
  booktitle = {2018 {{IEEE Symposium}} on {{Security}} and {{Privacy}} ({{SP}})},
  author = {Angel, Sebastian and Chen, Hao and Laine, Kim and Setty, Srinath},
  year = {2018},
  month = may,
  pages = {962--979},
  publisher = {IEEE},
  address = {San Francisco, CA},
  doi = {10.1109/SP.2018.00062},
  urldate = {2020-08-14},
  isbn = {978-1-5386-4353-2},
  langid = {english}
}

@techreport{barkerRecommendationKeyManagement2020,
  title = {Recommendation for Key Management: Part 1 - General},
  shorttitle = {Recommendation for Key Management},
  author = {Barker, Elaine},
  year = {2020},
  month = may,
  number = {NIST SP 800-57pt1r5},
  pages = {NIST SP 800-57pt1r5},
  address = {Gaithersburg, MD},
  institution = {{National Institute of Standards and Technology}},
  doi = {10.6028/NIST.SP.800-57pt1r5},
  urldate = {2024-04-11}
}

@inproceedings{beckRandomizedDecryptionRD2015,
  title = {Randomized Decryption ({{RD}}) Mode of Operation for Homomorphic Cryptography - Increasing Encryption, Communication and Storage Efficiency},
  booktitle = {2015 {{IEEE Conference}} on {{Computer Communications Workshops}} ({{INFOCOM WKSHPS}})},
  author = {Beck, Martin},
  year = {2015},
  month = apr,
  pages = {220--226},
  doi = {10.1109/INFCOMW.2015.7179388},
  urldate = {2023-10-12}
}

@inproceedings{benalohDenseProbabilisticEncryption1994,
  title = {Dense {{Probabilistic Encryption}}},
  booktitle = {Proceedings of the Workshop on Selected Areas of Cryptography},
  author = {Benaloh, Josh},
  year = {1994},
  pages = {120--128},
  langid = {english}
}

@inproceedings{brakerskiEfficientFullyHomomorphic2011,
  title = {Efficient {{Fully Homomorphic Encryption}} from ({{Standard}}) {{LWE}}},
  booktitle = {2011 {{IEEE}} 52nd {{Annual Symposium}} on {{Foundations}} of {{Computer Science}}},
  author = {Brakerski, Zvika and Vaikuntanathan, Vinod},
  year = {2011},
  month = oct,
  pages = {97--106},
  publisher = {IEEE},
  address = {Palm Springs, CA, USA},
  doi = {10.1109/FOCS.2011.12},
  urldate = {2022-07-27},
  isbn = {978-0-7695-4571-4 978-1-4577-1843-4},
  langid = {english}
}

@inproceedings{brakerskiFullyHomomorphicEncryption2012,
  title = {Fully {{Homomorphic Encryption}} without {{Modulus Switching}} from {{Classical GapSVP}}},
  booktitle = {Advances in {{Cryptology}} -- {{CRYPTO}} 2012},
  author = {Brakerski, Zvika},
  editor = {{Safavi-Naini}, Reihaneh and Canetti, Ran},
  year = {2012},
  series = {Lecture {{Notes}} in {{Computer Science}}},
  pages = {868--886},
  publisher = {Springer},
  address = {Berlin, Heidelberg},
  doi = {10.1007/978-3-642-32009-5_50},
  isbn = {978-3-642-32009-5},
  langid = {english}
}

@inproceedings{brakerskiLeveledFullyHomomorphic2012,
  title = {({{Leveled}}) Fully Homomorphic Encryption without Bootstrapping},
  booktitle = {Proceedings of the 3rd {{Innovations}} in {{Theoretical Computer Science Conference}}},
  author = {Brakerski, Zvika and Gentry, Craig and Vaikuntanathan, Vinod},
  year = {2012},
  month = jan,
  series = {{{ITCS}} '12},
  pages = {309--325},
  publisher = {Association for Computing Machinery},
  address = {New York, NY, USA},
  doi = {10.1145/2090236.2090262},
  urldate = {2023-09-22},
  isbn = {978-1-4503-1115-1}
}

@inproceedings{brakerskiLeveragingLinearDecryption2019,
  title = {Leveraging {{Linear Decryption}}: {{Rate-1 Fully-Homomorphic Encryption}} and {{Time-Lock Puzzles}}},
  shorttitle = {Leveraging {{Linear Decryption}}},
  booktitle = {Theory of {{Cryptography}}},
  author = {Brakerski, Zvika and D{\"o}ttling, Nico and Garg, Sanjam and Malavolta, Giulio},
  editor = {Hofheinz, Dennis and Rosen, Alon},
  year = {2019},
  series = {Lecture {{Notes}} in {{Computer Science}}},
  pages = {407--437},
  publisher = {Springer International Publishing},
  address = {Cham},
  doi = {10.1007/978-3-030-36033-7_16},
  isbn = {978-3-030-36033-7},
  langid = {english}
}

@article{canteautStreamCiphersPractical2018,
  title = {Stream {{Ciphers}}: {{A Practical Solution}} for {{Efficient Homomorphic-Ciphertext Compression}}},
  shorttitle = {Stream {{Ciphers}}},
  author = {Canteaut, Anne and Carpov, Sergiu and Fontaine, Caroline and Lepoint, Tancr{\`e}de and {Naya-Plasencia}, Mar{\'i}a and Paillier, Pascal and Sirdey, Renaud},
  year = {2018},
  month = jul,
  journal = {Journal of Cryptology},
  volume = {31},
  number = {3},
  pages = {885--916},
  issn = {1432-1378},
  doi = {10.1007/s00145-017-9273-9},
  urldate = {2023-10-10},
  langid = {english}
}

@incollection{chenEfficientHomomorphicConversion2021,
  title = {Efficient {{Homomorphic Conversion Between}} ({{Ring}}) {{LWE Ciphertexts}}},
  booktitle = {Applied {{Cryptography}} and {{Network Security}}},
  author = {Chen, Hao and Dai, Wei and Kim, Miran and Song, Yongsoo},
  editor = {Sako, Kazue and Tippenhauer, Nils Ole},
  year = {2021},
  volume = {12726},
  pages = {460--479},
  publisher = {Springer International Publishing},
  address = {Cham},
  doi = {10.1007/978-3-030-78372-3_18},
  urldate = {2023-09-13},
  isbn = {978-3-030-78371-6 978-3-030-78372-3},
  langid = {english}
}

@article{chenLabeledPSIFully,
  title = {Labeled {{PSI}} from {{Fully Homomorphic Encryption}} with {{Malicious Security}}},
  author = {Chen, Hao and Huang, Zhicong and Laine, Kim and Rindal, Peter},
  urldate = {2021-10-12}
}

@inproceedings{chenLabeledPSIFully2018,
  title = {Labeled {{PSI}} from {{Fully Homomorphic Encryption}} with {{Malicious Security}}},
  booktitle = {Proceedings of the 2018 {{ACM SIGSAC Conference}} on {{Computer}} and {{Communications Security}}},
  author = {Chen, Hao and Huang, Zhicong and Laine, Kim and Rindal, Peter},
  year = {2018},
  month = oct,
  pages = {1223--1237},
  publisher = {ACM},
  address = {Toronto Canada},
  doi = {10.1145/3243734.3243836},
  urldate = {2022-01-17},
  isbn = {978-1-4503-5693-0},
  langid = {english}
}

@inproceedings{cheonHomomorphicEncryptionArithmetic2017a,
  title = {Homomorphic {{Encryption}} for {{Arithmetic}} of {{Approximate Numbers}}},
  booktitle = {Advances in {{Cryptology}} -- {{ASIACRYPT}} 2017},
  author = {Cheon, Jung Hee and Kim, Andrey and Kim, Miran and Song, Yongsoo},
  editor = {Takagi, Tsuyoshi and Peyrin, Thomas},
  year = {2017},
  series = {Lecture {{Notes}} in {{Computer Science}}},
  pages = {409--437},
  publisher = {Springer International Publishing},
  address = {Cham},
  doi = {10.1007/978-3-319-70694-8_15},
  isbn = {978-3-319-70694-8},
  langid = {english}
}

@article{chillottiTFHEFastFully2020,
  title = {{{TFHE}}: {{Fast Fully Homomorphic Encryption Over}} the {{Torus}}},
  shorttitle = {{{TFHE}}},
  author = {Chillotti, Ilaria and Gama, Nicolas and Georgieva, Mariya and Izabach{\`e}ne, Malika},
  year = {2020},
  month = jan,
  journal = {Journal of Cryptology},
  volume = {33},
  number = {1},
  pages = {34--91},
  issn = {0933-2790, 1432-1378},
  doi = {10.1007/s00145-019-09319-x},
  urldate = {2022-02-25},
  langid = {english}
}

@article{chorPrivateInformationRetrieval1998a,
  title = {Private Information Retrieval},
  author = {Chor, Benny and Kushilevitz, Eyal and Goldreich, Oded and Sudan, Madhu},
  year = {1998},
  month = nov,
  journal = {Journal of the ACM},
  volume = {45},
  number = {6},
  pages = {965--981},
  issn = {0004-5411},
  doi = {10.1145/293347.293350},
  urldate = {2023-11-13}
}

@inproceedings{choTranscipheringFrameworkApproximate2021a,
  title = {Transciphering {{Framework}} for {{Approximate Homomorphic Encryption}}},
  booktitle = {Advances in {{Cryptology}} -- {{ASIACRYPT}} 2021},
  author = {Cho, Jihoon and Ha, Jincheol and Kim, Seongkwang and Lee, Byeonghak and Lee, Joohee and Lee, Jooyoung and Moon, Dukjae and Yoon, Hyojin},
  editor = {Tibouchi, Mehdi and Wang, Huaxiong},
  year = {2021},
  series = {Lecture {{Notes}} in {{Computer Science}}},
  pages = {640--669},
  publisher = {Springer International Publishing},
  address = {Cham},
  doi = {10.1007/978-3-030-92078-4_22},
  isbn = {978-3-030-92078-4},
  langid = {english}
}

@inproceedings{corrigan-gibbsSingleServerPrivateInformation2022,
  title = {Single-{{Server Private Information Retrieval}} with~{{Sublinear Amortized Time}}},
  booktitle = {Advances in {{Cryptology}} -- {{EUROCRYPT}} 2022: 41st {{Annual International Conference}} on the {{Theory}} and {{Applications}} of {{Cryptographic Techniques}}, {{Trondheim}}, {{Norway}}, {{May}} 30 -- {{June}} 3, 2022, {{Proceedings}}, {{Part II}}},
  author = {{Corrigan-Gibbs}, Henry and Henzinger, Alexandra and Kogan, Dmitry},
  year = {2022},
  month = may,
  pages = {3--33},
  publisher = {Springer-Verlag},
  address = {Berlin, Heidelberg},
  doi = {10.1007/978-3-031-07085-3_1},
  urldate = {2024-01-27},
  isbn = {978-3-031-07084-6}
}

@inproceedings{damgardGeneralisationSimplificationApplications2001a,
  title = {A {{Generalisation}}, a {{Simplification}} and {{Some Applications}} of {{Paillier}}'s {{Probabilistic Public-Key System}}},
  booktitle = {Public {{Key Cryptography}}},
  author = {Damg{\aa}rd, Ivan and Jurik, Mads},
  editor = {Kim, Kwangjo},
  year = {2001},
  series = {Lecture {{Notes}} in {{Computer Science}}},
  pages = {119--136},
  publisher = {Springer},
  address = {Berlin, Heidelberg},
  doi = {10.1007/3-540-44586-2_9},
  isbn = {978-3-540-44586-9},
  langid = {english}
}

@article{davidsonFrodoPIRSimpleScalable2023,
  title = {{{FrodoPIR}}: {{Simple}}, {{Scalable}}, {{Single-Server Private Information Retrieval}}},
  shorttitle = {{{FrodoPIR}}},
  author = {Davidson, Alex and Pestana, Gon{\c c}alo and Celi, Sof{\'i}a},
  year = {2023},
  month = jan,
  journal = {Proceedings on Privacy Enhancing Technologies},
  volume = {2023},
  number = {1},
  pages = {365--383},
  issn = {2299-0984},
  doi = {10.56553/popets-2023-0022},
  urldate = {2023-05-16},
  langid = {english}
}

@misc{dobraunigPastaCaseHybrid2021,
  title = {Pasta: {{A Case}} for {{Hybrid Homomorphic Encryption}}},
  shorttitle = {Pasta},
  author = {Dobraunig, Christoph and Grassi, Lorenzo and Helminger, Lukas and Rechberger, Christian and Schofnegger, Markus and Walch, Roman},
  year = {2021},
  number = {2021/731},
  urldate = {2023-09-13}
}

@incollection{dobraunigRastaCipherLow2018,
  title = {Rasta: {{A Cipher}} with {{Low ANDdepth}} and {{Few ANDs}} per {{Bit}}},
  shorttitle = {Rasta},
  booktitle = {Advances in {{Cryptology}} -- {{CRYPTO}} 2018},
  author = {Dobraunig, Christoph and Eichlseder, Maria and Grassi, Lorenzo and Lallemand, Virginie and Leander, Gregor and List, Eik and Mendel, Florian and Rechberger, Christian},
  editor = {Shacham, Hovav and Boldyreva, Alexandra},
  year = {2018},
  volume = {10991},
  pages = {662--692},
  publisher = {Springer International Publishing},
  address = {Cham},
  doi = {10.1007/978-3-319-96884-1_22},
  urldate = {2022-01-15},
  isbn = {978-3-319-96883-4 978-3-319-96884-1},
  langid = {english}
}

@inproceedings{ducasFHEWBootstrappingHomomorphic2015,
  title = {{{FHEW}}: {{Bootstrapping Homomorphic Encryption}} in {{Less Than}} a {{Second}}},
  shorttitle = {{{FHEW}}},
  booktitle = {Advances in {{Cryptology}} -- {{EUROCRYPT}} 2015},
  author = {Ducas, L{\'e}o and Micciancio, Daniele},
  editor = {Oswald, Elisabeth and Fischlin, Marc},
  year = {2015},
  series = {Lecture {{Notes}} in {{Computer Science}}},
  pages = {617--640},
  publisher = {Springer},
  address = {Berlin, Heidelberg},
  doi = {10.1007/978-3-662-46800-5_24},
  isbn = {978-3-662-46800-5},
  langid = {english}
}

@inproceedings{elgamalPublicKeyCryptosystem1985,
  title = {A {{Public Key Cryptosystem}} and a {{Signature Scheme Based}} on {{Discrete Logarithms}}},
  booktitle = {Advances in {{Cryptology}}},
  author = {ElGamal, Taher},
  editor = {Blakley, George Robert and Chaum, David},
  year = {1985},
  series = {Lecture {{Notes}} in {{Computer Science}}},
  pages = {10--18},
  publisher = {Springer},
  address = {Berlin, Heidelberg},
  doi = {10.1007/3-540-39568-7_2},
  isbn = {978-3-540-39568-3},
  langid = {english}
}

@article{fan2012somewhat,
  title = {Somewhat {{Practical Fully Homomorphic Encryption}}},
  author = {Fan, Junfeng and Vercauteren, Frederik},
  year = {2012},
  journal = {Proceedings of the 15th international conference on Practice and Theory in Public Key Cryptography},
  volume = {2012},
  pages = {1--16},
  publisher = {Citeseer}
}

@inproceedings{gentryCompressibleFHEApplications2019,
  title = {Compressible {{FHE}} with {{Applications}} to {{PIR}}},
  booktitle = {Theory of {{Cryptography}}},
  author = {Gentry, Craig and Halevi, Shai},
  editor = {Hofheinz, Dennis and Rosen, Alon},
  year = {2019},
  series = {Lecture {{Notes}} in {{Computer Science}}},
  pages = {438--464},
  publisher = {Springer International Publishing},
  address = {Cham},
  doi = {10.1007/978-3-030-36033-7_17},
  isbn = {978-3-030-36033-7},
  langid = {english}
}

@phdthesis{gentryFullyHomomorphicEncryption2009,
  title = {A Fully Homomorphic Encryption Scheme},
  author = {Gentry, Craig},
  year = {2009},
  address = {Stanford, CA, USA},
  school = {Stanford University}
}

@inproceedings{henzingerOneServerPrice2023,
  title = {One {{Server}} for the {{Price}} of {{Two}}: {{Simple}} and {{Fast}} \{\vphantom\}{{Single-Server}}\vphantom\{\} {{Private Information Retrieval}}},
  shorttitle = {One {{Server}} for the {{Price}} of {{Two}}},
  booktitle = {32nd {{USENIX Security Symposium}} ({{USENIX Security}} 23)},
  author = {Henzinger, Alexandra and Hong, Matthew M. and {Corrigan-Gibbs}, Henry and Meiklejohn, Sarah and Vaikuntanathan, Vinod},
  year = {2023},
  pages = {3889--3905},
  urldate = {2023-10-17},
  isbn = {978-1-939133-37-3},
  langid = {english}
}

@phdthesis{huImprovingEfficiencyHomomorphic2013,
  title = {Improving the {{Efficiency}} of {{Homomorphic Encryption Schemes}}},
  author = {Hu, Yin},
  year = {2013},
  langid = {english},
  school = {Worcester Polytechnic Institute}
}

@article{koganPrivateBlocklistLookupsa,
  title = {Private {{Blocklist Lookups}} with {{Checklist}}},
  author = {Kogan, Dmitry and {Corrigan-Gibbs}, Henry},
  urldate = {2021-10-09},
  isbn = {9781939133243}
}

@misc{liHintlessSingleServerPrivate2023,
  title = {Hintless {{Single-Server Private Information Retrieval}}},
  author = {Li, Baiyu and Micciancio, Daniele and Raykova, Mariana and {Schultz-Wu}, Mark},
  year = {2023},
  number = {2023/1733},
  urldate = {2024-04-21}
}

@inproceedings{liProtocolsCheckingCompromised2019a,
  title = {Protocols for {{Checking Compromised Credentials}}},
  booktitle = {Proceedings of the 2019 {{ACM SIGSAC Conference}} on {{Computer}} and {{Communications Security}}},
  author = {Li, Lucy and Pal, Bijeeta and Ali, Junade and Sullivan, Nick and Chatterjee, Rahul and Ristenpart, Thomas},
  year = {2019},
  month = nov,
  series = {{{CCS}} '19},
  pages = {1387--1403},
  publisher = {Association for Computing Machinery},
  address = {New York, NY, USA},
  doi = {10.1145/3319535.3354229},
  urldate = {2023-05-29},
  isbn = {978-1-4503-6747-9}
}

@inproceedings{lyubashevskyIdealLatticesLearning2010,
  title = {On {{Ideal Lattices}} and {{Learning}} with {{Errors}} over {{Rings}}},
  booktitle = {Advances in {{Cryptology}} -- {{EUROCRYPT}} 2010},
  author = {Lyubashevsky, Vadim and Peikert, Chris and Regev, Oded},
  editor = {Gilbert, Henri},
  year = {2010},
  series = {Lecture {{Notes}} in {{Computer Science}}},
  pages = {1--23},
  publisher = {Springer},
  address = {Berlin, Heidelberg},
  doi = {10.1007/978-3-642-13190-5_1},
  isbn = {978-3-642-13190-5},
  langid = {english}
}

@inproceedings{mahdaviConstantweightPIRSingleround2022,
  title = {Constant-Weight {{PIR}}: {{Single-round Keyword PIR}} via {{Constant-weight Equality Operators}}},
  shorttitle = {Constant-Weight \{\vphantom\}{{PIR}}\vphantom\{\}},
  booktitle = {31st {{USENIX Security Symposium}} ({{USENIX Security}} 22)},
  author = {Mahdavi, Rasoul Akhavan and Kerschbaum, Florian},
  year = {2022},
  pages = {1723--1740},
  urldate = {2023-05-12},
  isbn = {978-1-939133-31-1},
  langid = {english}
}

@misc{mahdaviPEPSIPracticallyEfficient2023,
  title = {{{PEPSI}}: {{Practically Efficient Private Set Intersection}} in the {{Unbalanced Setting}}},
  shorttitle = {{{PEPSI}}},
  author = {Mahdavi, Rasoul Akhavan and Lukas, Nils and Ebrahimianghazani, Faezeh and Humphries, Thomas and Kacsmar, Bailey and Premkumar, John and Li, Xinda and Oya, Simon and Amjadian, Ehsan and Kerschbaum, Florian},
  year = {2023},
  month = oct,
  number = {arXiv:2310.14565},
  eprint = {2310.14565},
  primaryclass = {cs},
  publisher = {arXiv},
  doi = {10.48550/arXiv.2310.14565},
  urldate = {2024-02-29},
  archiveprefix = {arxiv}
}

@inproceedings{menonSPIRALFastHighRate2022,
  title = {{{SPIRAL}}: {{Fast}}, {{High-Rate Single-Server PIR}} via {{FHE Composition}}},
  shorttitle = {{{SPIRAL}}},
  booktitle = {2022 {{IEEE Symposium}} on {{Security}} and {{Privacy}} ({{SP}})},
  author = {Menon, Samir Jordan and Wu, David J.},
  year = {2022},
  month = may,
  pages = {930--947},
  publisher = {IEEE},
  address = {San Francisco, CA, USA},
  doi = {10.1109/SP46214.2022.9833700},
  urldate = {2023-05-12},
  isbn = {978-1-66541-316-9},
  langid = {english}
}

@misc{menonYPIRHighThroughputSingleServer2024,
  title = {{{YPIR}}: {{High-Throughput Single-Server PIR}} with {{Silent Preprocessing}}},
  shorttitle = {{{YPIR}}},
  author = {Menon, Samir Jordan and Wu, David J.},
  year = {2024},
  number = {2024/270},
  urldate = {2024-04-21}
}

@misc{MicrosoftSEALRelease2023,
  title = {Microsoft {{SEAL}} (Release 4.1)},
  year = {2023},
  month = jan
}

@inproceedings{mugheesOnionPIRResponseEfficient2021,
  title = {{{OnionPIR}}: {{Response Efficient Single-Server PIR}}},
  shorttitle = {{{OnionPIR}}},
  booktitle = {Proceedings of the 2021 {{ACM SIGSAC Conference}} on {{Computer}} and {{Communications Security}}},
  author = {Mughees, Muhammad Haris and Chen, Hao and Ren, Ling},
  year = {2021},
  month = nov,
  series = {{{CCS}} '21},
  pages = {2292--2306},
  publisher = {Association for Computing Machinery},
  address = {New York, NY, USA},
  doi = {10.1145/3460120.3485381},
  urldate = {2023-06-16},
  isbn = {978-1-4503-8454-4}
}

@inproceedings{naehrigCanHomomorphicEncryption2011a,
  title = {Can Homomorphic Encryption Be Practical?},
  booktitle = {Proceedings of the 3rd {{ACM}} Workshop on {{Cloud}} Computing Security Workshop},
  author = {Naehrig, Michael and Lauter, Kristin and Vaikuntanathan, Vinod},
  year = {2011},
  month = oct,
  series = {{{CCSW}} '11},
  pages = {113--124},
  publisher = {Association for Computing Machinery},
  address = {New York, NY, USA},
  doi = {10.1145/2046660.2046682},
  urldate = {2023-10-11},
  isbn = {978-1-4503-1004-8}
}

@inproceedings{paillierPublicKeyCryptosystemsBased1999,
  title = {Public-{{Key Cryptosystems Based}} on {{Composite Degree Residuosity Classes}}},
  booktitle = {Advances in {{Cryptology}} --- {{EUROCRYPT}} '99},
  author = {Paillier, Pascal},
  editor = {Stern, Jacques},
  year = {1999},
  series = {Lecture {{Notes}} in {{Computer Science}}},
  pages = {223--238},
  publisher = {Springer},
  address = {Berlin, Heidelberg},
  doi = {10.1007/3-540-48910-X_16},
  isbn = {978-3-540-48910-8},
  langid = {english}
}

@inproceedings{patelDonBeDense2023,
  title = {\{\vphantom\}{{Don}}'t\vphantom\{\} Be {{Dense}}: {{Efficient Keyword}} \{\vphantom\}{{PIR}}\vphantom\{\} for {{Sparse Databases}}},
  shorttitle = {\{\vphantom\}{{Don}}'t\vphantom\{\} Be {{Dense}}},
  booktitle = {32nd {{USENIX Security Symposium}} ({{USENIX Security}} 23)},
  author = {Patel, Sarvar and Seo, Joon Young and Yeo, Kevin},
  year = {2023},
  pages = {3853--3870},
  urldate = {2024-04-21},
  isbn = {978-1-939133-37-3},
  langid = {english}
}

@article{patelPrivateStatefulInformation2018,
  title = {Private {{Stateful Infor-mation Retrieval}}},
  author = {Patel, Sarvar and Persiano, Giuseppe and Yeo, Kevin},
  year = {2018},
  journal = {ACM SIGSAC Conference on Computer \& Communications Security},
  volume = {18},
  pages = {18},
  publisher = {ACM},
  doi = {10.1145/3243734.3243821},
  urldate = {2021-10-18},
  isbn = {9781450356930}
}

@misc{PrivateStatefulInformation,
  title = {Private {{Stateful Information Retrieval}} {\textbar} {{Proceedings}} of the 2018 {{ACM SIGSAC Conference}} on {{Computer}} and {{Communications Security}}},
  urldate = {2023-09-19},
  howpublished = {https://dl.acm.org/doi/10.1145/3243734.3243821}
}

@article{regevLatticesLearningErrors2009,
  title = {On Lattices, Learning with Errors, Random Linear Codes, and Cryptography},
  author = {Regev, Oded},
  year = {2009},
  month = sep,
  journal = {Journal of the ACM},
  volume = {56},
  number = {6},
  pages = {1--40},
  issn = {0004-5411, 1557-735X},
  doi = {10.1145/1568318.1568324},
  urldate = {2023-09-28},
  langid = {english}
}

@inproceedings{Thomas2019ProtectingAF,
  title = {Protecting Accounts from Credential Stuffing with Password Breach Alerting},
  booktitle = {{{USENIX Security Symposium}}},
  author = {Thomas, Kurt and Pullman, Jennifer and Yeo, Kevin and Raghunathan, A and Kelley, Patrick Gage and Invernizzi, L and Benko, B and Pietraszek, Tadek and Patel, S and Boneh, D and Bursztein, Elie},
  year = {2019}
}

@misc{zamaConcreteTFHECompiler2022,
  title = {Concrete: {{TFHE Compiler}} That Converts Python Programs into {{FHE}} Equivalent},
  author = {Zama},
  year = {2022}
}

@INPROCEEDINGS{piano,
  author={Zhou, Mingxun and Park, Andrew and Zheng, Wenting and Shi, Elaine},
  booktitle={2024 IEEE Symposium on Security and Privacy (SP)}, 
  title={Piano: Extremely Simple, Single-Server PIR with Sublinear Server Computation}, 
  year={2024},
  volume={},
  number={},
  pages={4296-4314},
  doi={10.1109/SP54263.2024.00055}}

@inproceedings{thorpir,
  author = {Fisch, Ben and Lazzaretti, Arthur and Liu, Zeyu and Papamanthou, Charalampos},
  title = {ThorPIR: Single Server PIR via Homomorphic Thorp Shuffles},
  year = {2024},
  isbn = {9798400706363},
  publisher = {Association for Computing Machinery},
  address = {New York, NY, USA},
  url = {https://doi.org/10.1145/3658644.3690326},
  doi = {10.1145/3658644.3690326},
  booktitle = {Proceedings of the 2024 on ACM SIGSAC Conference on Computer and Communications Security},
  pages = {1448–1462},
  numpages = {15},
  location = {Salt Lake City, UT, USA},
  series = {CCS '24}
}

@inproceedings{10.1145/3658644.3690266,
  author = {Ren, Ling and Mughees, Muhammad Haris and Sun, I},
  title = {Simple and Practical Amortized Sublinear Private Information Retrieval using Dummy Subsets},
  year = {2024},
  isbn = {9798400706363},
  publisher = {Association for Computing Machinery},
  address = {New York, NY, USA},
  url = {https://doi.org/10.1145/3658644.3690266},
  doi = {10.1145/3658644.3690266},
  booktitle = {Proceedings of the 2024 on ACM SIGSAC Conference on Computer and Communications Security},
  pages = {1420–1433},
  numpages = {14},
  location = {Salt Lake City, UT, USA},
  series = {CCS '24}
}

@inproceedings{10.1145/3658644.3690328,
    author = {Burton, Alexander and Menon, Samir Jordan and Wu, David J.},
    title = {Respire: High-Rate PIR for Databases with Small Records},
    year = {2024},
    isbn = {9798400706363},
    publisher = {Association for Computing Machinery},
    address = {New York, NY, USA},
    url = {https://doi.org/10.1145/3658644.3690328},
    doi = {10.1145/3658644.3690328},
    booktitle = {Proceedings of the 2024 on ACM SIGSAC Conference on Computer and Communications Security},
    pages = {1463–1477},
    numpages = {15},
    location = {Salt Lake City, UT, USA},
    series = {CCS '24}
}

@INPROCEEDINGS {mazmudar2024peer2pirprivatequeriesipfs,
    author = { Mazmudar, Miti and Veitch, Shannon and Mahdavi, Rasoul Akhavan },
    booktitle = { 2025 IEEE Symposium on Security and Privacy (SP) },
    title = {{ Peer2PIR: Private Queries for IPFS }},
    year = {2025},
    volume = {},
    ISSN = {},
    pages = {4438-4456},
    doi = {10.1109/SP61157.2025.00231},
    url = {https://doi.ieeecomputersociety.org/10.1109/SP61157.2025.00231},
    publisher = {IEEE Computer Society},
    address = {Los Alamitos, CA, USA},
    month =May
}

@misc{optOutSCTDeBlasio,
      title={Opt-out SCT auditing in Chrome.}, 
      author={Joe DeBlasio},
      url={https://docs.google.com/document/d/16GQ7iN3kB46GSW5b-sfH5MO3nKSYyEb77YsM7TMZGE/edit.}, 
}

@misc{lattigo,
    title = {Lattigo v6},
    howpublished = {Online: \url{https://github.com/tuneinsight/lattigo}},
    month = Aug,
    year = 2024,
    note = {EPFL-LDS, Tune Insight SA}
}

@inproceedings{10.1007/978-3-540-85174-5_31,
author = {Peikert, Chris and Vaikuntanathan, Vinod and Waters, Brent},
title = {A Framework for Efficient and Composable Oblivious Transfer},
year = {2008},
isbn = {9783540851738},
publisher = {Springer-Verlag},
address = {Berlin, Heidelberg},
url = {https://doi.org/10.1007/978-3-540-85174-5_31},
doi = {10.1007/978-3-540-85174-5_31},
booktitle = {Proceedings of the 28th Annual Conference on Cryptology: Advances in Cryptology},
pages = {554–571},
numpages = {18},
location = {Santa Barbara, CA, USA},
series = {CRYPTO 2008}
}

@inproceedings{10.1145/2213977.2214086,
author = {L\'{o}pez-Alt, Adriana and Tromer, Eran and Vaikuntanathan, Vinod},
title = {On-the-fly multiparty computation on the cloud via multikey fully homomorphic encryption},
year = {2012},
isbn = {9781450312455},
publisher = {Association for Computing Machinery},
address = {New York, NY, USA},
url = {https://doi.org/10.1145/2213977.2214086},
doi = {10.1145/2213977.2214086},
booktitle = {Proceedings of the Forty-Fourth Annual ACM Symposium on Theory of Computing},
pages = {1219–1234},
numpages = {16},
location = {New York, New York, USA},
series = {STOC '12}
}
